\documentclass{article}
\usepackage[final]{style} 
\usepackage{natbib}
\usepackage[utf8]{inputenc} 
\usepackage[T1]{fontenc}    
\usepackage[colorlinks=true,citecolor=blue]{hyperref}       
\usepackage{url}            
\usepackage{booktabs}       
\usepackage{authblk}
\usepackage{amsfonts}       
\usepackage{nicefrac}       
\usepackage{microtype}      
\usepackage{dutchcal}
\usepackage{xcolor}         
\usepackage{wrapfig}       
\usepackage{mathtools} 
\usepackage{amssymb}
\usepackage{aligned-overset}
\usepackage[bbgreekl]{mathbbol}

\usepackage{tikz} 
\usepackage{algpseudocode} 
\usepackage[linesnumbered]{algorithm2e}
\usepackage{amsthm}
\usepackage{cuted}
\usepackage{ amssymb }
\usepackage{amsmath}
\usepackage{mathtools, stmaryrd}
\usepackage{subcaption}
\usepackage{bbm}
\usepackage[bbgreekl]{mathbbol}
\usepackage[inline]{enumitem}
\usepackage[font=scriptsize,labelfont=bf]{caption}
\DeclareMathOperator*{\argmax}{argmax} 
\DeclareMathOperator*{\argmin}{argmin}
\usepackage{mdframed}
\newmdtheoremenv{theorem}{Theorem}
\newtheorem{remark}{Remark}

\newmdtheoremenv{definition}{Definition}

\newmdtheoremenv{proposition}{Proposition}[theorem]
\newcommand{\RN}[1]{%
\textup{\uppercase\expandafter{\romannumeral#1}}%
}
\newenvironment{proof_method} {\begin{proof}[Proof method]} {\end{proof}}
\renewcommand{\O}[1]{\ensuremath{\mathsf{#1}}} 

\usepackage{tikz}

\providecommand{\keywords}[1]
{
  \small	
  \textbf{\textit{Keywords---}} #1
}

\title{Randomized Transport Plans via Hierarchical Fully Probabilistic Design}
\author{%
Sarah Boufelja Y.$^{1}$\thanks{Corresponding author.} \quad Anthony Quinn$^{1,2}$ \quad Robert Shorten$^1$  \\
$^1$Imperial College London, Dyson School of Design Engineering \\ $^2$Trinity College Dublin, School of Engineering \\
\texttt{\{s.boufelja21,a.quinn,r.shorten\}@imperial.ac.uk}
}

\begin{document}
\maketitle
\section*{Abstract}
An optimal {\em randomized\/} strategy for design of balanced, normalized mass transport plans is developed. It replaces---but specializes to---the {\em deterministic\/}, regularized optimal transport (OT) strategy, which yields only a certainty-equivalent plan. The incompletely specified---and therefore uncertain---transport plan is acknowledged to be a random process. Therefore, hierarchical fully probabilistic design (HFPD) is adopted, yielding  an optimal hyperprior supported on the set of possible transport plans, and consistent with prior mean constraints on the marginals of the uncertain plan.   This Bayesian resetting of the design problem for transport plans ---which we call HFPD-OT---confers new opportunities. These include (i) a strategy for the  generation of a random sample of {\em joint\/} transport plans; (ii)  randomized {\em marginal contracts\/} for individual source-target pairs; and (iii) consistent measures of uncertainty in the plan and its contracts.  An application in  fair market matching is outlined, in which HFPD-OT enables the recruitment of a more diverse subset  of contracts---than is possible in  classical OT---into the delivery of an expected plan.  

\keywords{\textit{Optimal transport, Bayesian hierarchical modelling, Fully probabilistic design, Convex optimization,  Algorithmic fairness, Market matching}}

\section{Main Contributions}\label{sec:maincontribs}

\noindent
Optimal transport (OT) refers to the classical design of a {\em deterministic\/} transport  plan, $\pi$, for taking a unit\footnote{Throughout this paper, we address only the balanced, normalized transport problem.} mass—distributed across a source domain, $\mathbb{\mathbb{\Omega}}_{X}$—and redistributing it across a target domain, $\mathbb{\mathbb{\Omega}}_{Y}$. The transport plan is expressed as an unknown, deterministic, joint distribution, $\pi$, with support in  $\mathbb{\mathbb{\Omega}}_{X} \times \mathbb{\mathbb{\Omega}}_{Y}$. The distributed source and target are therefore the marginals of $\pi$, and are specified {\em a priori\/} by  $\mu_0$ and $\nu_0$ on $\mathbb{\mathbb{\Omega}}_{X}$ and $\mathbb{\mathbb{\Omega}}_{Y}$, respectively. Consequently, $\pi$ is confined to the space, $\mathbb{\Pi}(\mu_0,\nu_0)$, of distributions on $\mathbb{\mathbb{\Omega}}_{X} \times \mathbb{\mathbb{\Omega}}_{Y}$, with $\mu_0$ and $\nu_0$ as its marginals. An optimal choice, $\pi^o$, of $\pi$—called the OT plan—is achieved by minimizing the expected value, under $\pi$, of a pre-specified cost of transport, ${\mathsf C}(x,y)$, from $\mathbb{\mathbb{\Omega}}_{X}$ to  $ \mathbb{\mathbb{\Omega}}_{Y}$. 

\noindent
In this paper, we reformulate the design of transport maps in the {\em Bayesian\/} (i.e.\ fully probabilistic) way. In particular, deterministic optimization---yielding $\pi^o$—is replaced by the hierarchical fully probabilistic design (HFPD) of an optimal  {\em randomized decision-making strategy}, $\pi \sim  \mathsf{S}^o$ (i.e.\ a hyperprior), for choosing $\pi $. This approach recognizes that the unknown transport plan, $\pi$, is a (generally nonparametric) random 
process. We therefore  equip it with a prior, $\mathsf{S} (\pi | K)$, where $K$ denotes marginal (mean) knowledge constraints which will be detailed in the sequel. Following the axioms of FPD at this hierarchical level (i.e.\ HFPD), we equip the space, ${\mathbb S}_K$, of $\O{S}$---being the randomized strategy for choosing  the transport plan, $\pi$---with an appropriately formulated  loss function, and we minimize the expected value of the latter under $\mathsf{S}$. This yields the optimal randomized strategy, $\mathsf{S}^o (\pi | K)$, for choosing $\pi$, being  also the optimal  hyperprior for uncertain $\pi$. We show that  this procedure is equivalent to minimization of a Kullback-Leibler divergence (KLD), leading to a Gibbs form for   $\mathsf{S}^o (\pi | K)$:
\begin{equation}\label{eq:hyperprior_init}
\mathsf{S}^o (\pi | K) \propto \mathsf{S}_{\O{I}}(\pi | K) 
e^{-{\mathsf{D}_{\mathsf{KL}}}(\pi || \pi_{\O{I}})}
 e^{-\lambda_1^o {\mathsf{D}_{\mathsf{KL}}}(\mu || \mu_0)}
e^{-\lambda_2^o {\mathsf{D}_{\mathsf{KL}}}(\nu || \nu_0)} \in {\mathbb S}_K.
\end{equation}
Here, $\mu$ and $\nu$ are the uncertain (i.e.\ random)  marginals of the random transport plan, $\pi$. The KLDs, ${\mathsf{D}_{\mathsf{KL}}}(\cdot || \cdot)$, act as Gibbs energies.  Meanwhile, $\mathsf{S}_{\O{I}}$ and $\pi_{\O{I}}$ are the freely but necessarily {\em pre-specified\/} zero-loss choices of $\mathsf{S}$ and $\pi$, respectively, referred to as the {\em ideal\/} or target choices. 

\noindent
By resetting OT as a problem of Bayesian decision making via HFPD, we achieve the following principal goals:
\begin{itemize}
\item[(i)] The deterministic, regularized  OT choice, $\pi^o$, obtained via constrained optimization at the base level of modelling, $(x,y) \sim \pi$, is replaced by an optimal generator of randomized plans (i.e.\ a randomized strategy for designing transport plans, $\pi$) at the hierarchical level of complete modelling, $\pi \sim \mathsf{S}^o (\pi | K)$.
 \item[(ii)] In the parametric case, in which the support set, ${\mathbb{\mathbb{\Omega}}_{X}}\times {\mathbb{\mathbb{\Omega}}_{Y}}$, is finite, we can compute optimal univariate (marginal and/or conditional) distributions, $\pi_{i,j} \sim \mathsf{S}^o_{i,j}$, for  modelling and randomization of the transport {\em contract}, $\pi_{ij}\in(0,1)$,  from the {\em agent\/} (at) $x_i$ to the agent (at) $y_j$.
 \item[(iii)] In line with all Bayesian decision-making strategies, we can summarize $\pi\sim\mathsf{S}^o (\pi | K)$ via a certainty-equivalent (CE) transport plan, $\hat{\pi} \in \mathbb{\Pi}_K$---such as its expected or maximally probable value---and equip this with a summary of our uncertainty in $\pi$ (e.g.\ via the Bayesian standard intervals for the contracts, $\pi_{i,j} \sim \mathsf{S}^o_{i,j}$). 
\end{itemize}
By equipping transport plans with an optimal hyperprior from which candidate plans can be generated, we are able to encode our prior knowledge and our ranking of preferences. This HFPD resetting of OT can have significant impact in applications. We consider one such  application, in algorithmic fairness. Specifically, we address the problem of labour market matching, in which fairness is induced  by optimally randomizing the matching strategy (a transport plan) via HFPD-OT, thereby increasing a diversity index among contracts. 

\section{Introduction to transport plan design and optimal transport}\label{intro}
Optimal Transport (OT) techniques have received increasing attention in the past decade, in a wide range of domains such as machine learning and generative adversarial learning \citep{arjovsky2017wasserstein}, domain adaptation \citep{OT_Domain_Adaptation}, image processing and watermarking \citep{OT_Watermarking}, hallucination detection in neural translation machines \citep{guerreiro2023optimal}, {\em etc}. In addition to traditional applications in economics and market matching \citep{galichon}, fluid mechanics and diffusion processes \citep{OT_fluid_dynamics}, it has also been used to perform sampling and Bayesian inference \citep{ELMOSELHY20127815}. 
\\\\ 
OT is concerned with the least costly transport plan (in expectation) between a source and a target probability measure. The unregularized OT plan induces a natural distance in the space of probability measures (the Kantorovitch-Rubinstein distance) \citep{Villani2008OptimalTO}, introducing a rich topological structure by lifting key geometric properties associated with the ground metric to the space of probability measures \citep{Villani2008OptimalTO,MAL-073}. For example, if the ground space is Euclidean, concepts like gradient, barycentre and convexity are naturally extended to the space of probability measures. 

Notwithstanding the wide range of applications, the classical formalism of OT confines it to a purely deterministic setting, which regards the transport plan as a crisp object and assumes perfect knowledge of the marginals (Figure \ref{fig:fig21}). It fails to model and (critically) translate uncertainty in the marginals to uncertainty in the design of transport plans. In this regard, classical OT is an instance of certainty-equivalence (CE) decision making, which produces myopic transport strategies that do not account for the uncertain and random nature of many real systems. One might think that recasting the classical OT problem in terms of robust optimization might address 
 these  issues. A robust optimization formulation relies on a deterministic, unknown but bounded description of the uncertainty in the marginals \citep{robust_optimization}. 
 Such a  design choice may be overly conservative: it indeed considers all possible outcomes in the uncertainty set, but may assign non-negligible weights even to plans that are highly improbable. Furthermore, the robust design is not equipped with a quantifier of the intrinsic uncertainty of the transport plan.  \\\\
In this manuscript, we propose the HFPD-OT  approach to the design of uncertain transport plans. It departs from the conventional OT setting by considering the transport plan as a random process. Consistent hierarchical Bayesian modelling endows the uncertain plan with its own {\em hyperprior} (Figure \ref{fig:fig22}). Its optimal choice provides a {\em randomized strategy\/} for choosing transport plans in the space of plans consistent with prior-imposed knowledge constraints on its marginals. It also acts as  a generative model for  random sampling of  transport plans.  By treating transport plans as random processes, we effectively recast the transport design problem as one of inference. This contrasts with the OT literature, which is only concerned with deterministic optimization strategies for choosing deterministic plans. As we will see in the literature review, next, the tools provided by HFPD-OT---intended for modelling and reasoning about uncertainty in  transport plans---are  not available in the classical OT setting.


\subsection{Approaches to modelling uncertainty in OT}\label{lit_review}
There are precedents in eliciting and processing uncertainty in OT, but they are generally couched in terms of base-level modelling, and not in terms of the hierarchical Bayesian approach developed here. Specifically, \begin{enumerate*} [label=(\roman*)] 
\item our  method is primarily concerned with the design of a fully probabilistic model over the space of transport plans; \item as such, the transport plan is modeled as a (generally nonparametric) process endowed with its own (hyper)prior;  and \item we rely on randomization techniques for choosing plans, in contrast to  existing methods which are mainly based on deterministic optimization techniques. 
\end{enumerate*}

Copulas \citep{sklar} are historically among the first methods proposed for the design of multivariate distributions with arbitrary, but perfectly known marginals. Other techniques relaxed this assumption to address situations where exactly one marginal is uncertain. This is the case in \citep{Goodman1953EcologicalRA}, for instance, where the authors model the uncertainty in one marginal with a Gaussian noise. In ecological inference (a case of parametric transport design on a finite support), \citep{wakefield_ecological_inference} studied the case where one marginal is uncertain, adopting a hierarchical multinomial-Dirichlet-based model.  We highlight two distinctions in our work: \begin{enumerate*} [label=(\roman*)] \item we do not impose a parametric constraint in general, and we allow for uncertainty in both marginals; and  \item the authors of the earlier paper pursue markedly different statistical inference objectives from OT \end{enumerate*}. 
\\\\
Interestingly, the connection between ecological inference and OT was not established until later, in \citep{pmlr-v97-frogner19a}, where the authors extended the previous model and studied the case where both marginals are uncertain. The questions we address in this paper again differ from those in  \citep{pmlr-v97-frogner19a}  in the following ways: \begin{enumerate*} [label=(\roman*)] \item they solve a base-level MAP optimization problem using a Bregman projection method, once again recovering a certainty-equivalent OT plan, whereas our primary goal is to depart from such a certainty-equivalence setting and design an optimal hierarchical Bayesian model from which random transport plans can be generated and used {\em in lieu\/} of an OT plan. If required (as we will see), the {\em expected\/} plan takes the place of the MAP plan as the Bayesian minimum-risk decision (i.e.\ estimate) of the uncertain plan, with asymptotic convergence  to the  MAP plan; and 
\item the derivations in \citep{pmlr-v97-frogner19a} rely on parametric and structural assumptions, mainly full separability. Separability is a strong assumption in that it excludes the modelling of  rich structures and interactions that may exist in real-world data. We do not require these assumptions in our hyperprior, and we leave it to the designer---via the specification of ideal designs (to be explained in the sequel)---to impose any relevant structural requirements. \end{enumerate*} 
\\\\
Uncertainty in the cost matrix in the finite case is considered by \citep{pmlr-v157-mallasto21a}. Given a finite sample of these cost matrices, they model the induced  uncertainty in the (finitely supported) OT plan. They do not allow for any uncertainty in the marginals, and so their distribution over OT plans is geometrically constrained to the OT polytope. They impose various standard parametric priors over this set, without any optimality claims for them. Our work significantly extends this treatment by modelling uncertainty in the marginals, so that  our hierarchical model has support in the geometrically {\em unconstrained\/} space of  transport plans, and  extends to the nonparametric setting of continuously  supported plans. Importantly---and in contrast to \citep{pmlr-v157-mallasto21a}---we do not impose an optimality constraint on the base-level plans themselves, but, rather, on the hierarchical (generative) distribution of (all possible) plans, $\O{S}^o(\pi | K)$  (\ref{eq:hyperprior_init}). In this way, the random generator of the plans, $\O{S}^o(\pi | K)$, is optimal, and not the uncertain transport plan, $\pi$, itself (although subsequent projections of $\O{S}^o (\pi | K )$ can yield optimal Bayesian decisions about $\pi$, in the conventional manner of Bayesian decision-making). The main contribution of our work is to {\em deduce\/} this optimal hyperprior for transport design (\ref{eq:hyperprior_init}) via the foundational methods of fully probabilistic design (FPD) \citep{KARNY2012105}. We show how this HFPD-OT hyperprior concentrates to the classical regularized OT solution as uncertainty in the marginals diminishes (\ref{eq:base_level_eot}). 
\\\\ 
An interesting line of work on unbalanced OT (UOT) in \citep{SEJOURNE2023407} relaxes the strict marginal constraints $\mathbb{\Pi}(\mu_0,\nu_0)$, and replaces them by a soft penalization, using Kullback-Leibler balls centred on the nominal marginals (as we do in this paper). This ensures feasibility of the UOT problem, allowing transport between unequal (non-probability) measures (which we do not allow in our work). Once again, their solution involves a base-level deterministic optimization.  \\\\ 
Finally, entropy-regularized OT (EOT) \citep{cuturi2013sinkhorn} is a foundational work on deterministic OT that will be recovered asymptotically via HFPD-OT. In EOT,  the classical OT linear program is relaxed by means of an entropy regularization term, yielding a strictly convex problem, which is amenable to efficient matrix scaling algorithms, notably Sinkhorn-Knopp~\citep{cuturi2013sinkhorn}. In our own recent paper \citep{sby2022fully},  we formally establish the relationship between  base-level EOT under the usual deterministic marginal constraints---therefore yielding a certainty-equivalent (i.e.\ singular) OT plan, $\pi^o$, in the conventional manner---and fully probabilistic design (FPD). In this paper, our goal  is to extend the base-level EOT setting by deriving an optimal hyperprior, $\O{S}^o(\pi | K)$ (\ref{eq:hyperprior_init}), over the set of uncertain transport plans.

\subsection{Notational conventions, technical preliminaries for non-hierarchical OT, and outline of the paper }\label{conventions}
In the following, we will review the key mathematical conventions used throughout the paper. Specifically, all probability measures will be referred to as (probability) distributions. The context will make clear whether the distribution in question is a probability density function (pdf) or a probability mass function (pmf). A superscript $o$ refers to \textit{optimal} distributions, e.g. $\mathsf{S}^{o}$, whereas a subscript $\mathsf{I}$ designates \textit{ideal} distributions, e.g. $\mathsf{S}_{\mathsf{I}}$. Moreover, all fixed and prior-elicited quantities are referred to using a subscript $0$ ($\mu_{0}$, $\nu_{0}$, {\em etc}.). Sets will be denoted by a blackboard  typeface (e.g. $\mathbb{\Omega}_{X}, \mathbb{\Omega}_{Y}, \mathbb{M}$, {\em etc}.), and deterministic functionals will be denoted by a math {\em sans serif\/} typeface (e.g. $\O{S}$, $\O{C}$, $\O{D}$, {\em etc}.). Instantiated distributions will be assigned a math calligraphic typeface ($\mathcal{U}$, $\mathcal{N}$, {\em etc}).
\begin{itemize}
\item The conventional non-hierarchical---which we call the base-level---probability space (triple) is ($\mathbb{\mathbb{\Omega}}$, $\mathcal{F}$, $\mathcal{P}$), where $\mathbb{\mathbb{\Omega}}$ is the sample space, $\mathcal{F}$ denotes the ($\sigma$-)algebra of measurable subsets of $\mathbb{\mathbb{\Omega}}$, and $\mathcal{P}$ is a probability measure defined on $\mathcal{F}$. 
\item Consider two random variables (rvs), 
$X$: $\mathbb{\mathbb{\Omega}} \mapsto \mathbb{\mathbb{\Omega}}_{X}$ and $Y$: $\mathbb{\mathbb{\Omega}} \mapsto \mathbb{\mathbb{\Omega}}_{Y}$, whose images, $\mathbb{\mathbb{\Omega}}_{X}$ and $\mathbb{\mathbb{\Omega}}_{Y}$, are, respectively, compact subsets of topological spaces of unspecified dimensions. In the standard setting of optimal transport (OT)~\citep{Villani2008OptimalTO,MAL-073}, their  marginal distributions under $(\mathbb{\mathbb{\Omega}}, \mathcal{F}, \mathcal{P})$ are prior-specified (i.e.\ {\em known\/}) to be $\mu_0 \in \mathbbm{P}(\mathbb{\mathbb{\Omega}}_{X})$ and $\nu_0 \in \mathbbm{P}(\mathbb{\Omega}_{Y})$, respectively, while their joint distribution,  $\pi \in \mathbbm{P}(\mathbb{\mathbb{\Omega}}_{X} \times \mathbb{\mathbb{\Omega}}_{Y})$, is {\em unknown}, and is the subject of design. 
\item The reference measure in $(\mathbb{\mathbb{\Omega}}_{X} \times \mathbb{\mathbb{\Omega}}_{Y})$ is denoted by $\lambda (x,y)$. Depending on the context, $\lambda$  interchangeably denotes the Lebesgue measure (in the continuous case) or the counting measure (in the discrete case). $\pi$, $\mu_0$ and $\nu_0$ are  absolutely continuous \textit{w.r.t.}\ $\lambda$.  We do not distinguish notationally between a probability measure and its Radon-Nikodym derivative \textit{w.r.t.}\ to $\lambda$, e.g.\  $\frac{d\pi}{d\lambda} \equiv\pi$, {\em etc}., and we refer to all as distributions.  
\item The prior-specified marginal constraints, $\mu_0$ and $\nu_0$, constrain $\pi$ to the following knowledge-constrained set:
\[
\label{eq:Kset}
\mathbb{\Pi}(\mu_0, \nu_0) \equiv\Bigl\{\pi \in \mathbbm{P}(\mathbb{\Omega}_{X} \times \mathbb{\Omega}_{Y}) \;|\; \int_{\mathbb{\Omega}_{y}}\pi d\lambda(y) \equiv\mu_0, \;\;
\int_{\mathbb{\Omega}_{X}}\pi d\lambda(x) \equiv\nu_0
\Bigr\}
\]
\item Consider an alternative distribution, $\zeta \in \mathbbm{P}(\mathbb{\Omega}_{X} \times \mathbb{\Omega}_{Y})$.  The Kullback-Leibler divergence (KLD) of $\zeta$ to $\pi$ is:
\begin{equation}
\mathsf{D}_{\mathsf{KL}}(\pi|| \zeta) \equiv
\left\{
\begin{aligned}
\int_{\mathbb{\Omega}_{X}\times\mathbb{\Omega}_{Y}} \pi(x,y)\log\Bigl(\frac{\pi(x,y)}{\zeta(x,y)}\Bigl)d\lambda(x,y) & \;\; \text{if} \;\; \pi \ll \zeta,  \\
+ \infty & \;\; \text{otherwise,} \\
\end{aligned} \right.
\end{equation}
where $\pi \ll \zeta$ indicates the absolute continuity (a.c.) of $\pi$ \textit{w.r.t.}\ $\zeta$.
\item If $\O{q}$ is an integrable function with domain $\mathbb{\Omega}_{X} \times \mathbb{\Omega}_{Y}$, then its expectation \textit{w.r.t} $\pi$ is defined as 
\[
\O{E}_{\pi}\left[\O{q}\right] \equiv\int_{\mathbb{\Omega}_{X} \times \mathbb{\Omega}_{Y}} \O{q}(x,y)\pi(x,y)d\lambda(x,y) \; < \infty
\]
\item 
 $K$---in, for example, \ $\O{S}(\pi | K)$---is  Jeffreys' notation~\citep{Jeffreys1939-JEFTOP-5}, encoding the knowledge which acts as a condition on a probability model. It effectively confines $\O{S}$ to a particular knowledge-constrained set, ${\mathbb S}_K$. Its specific meaning will be defined in context, at both the base level and hierarchical level, as appropriate.
 \item $\O{supp(\mu)}$ denotes the support of the distribution, $\mu$.
\item $<\!\!\cdot,\!\cdot\!\!>$ denotes the standard inner product between vectors in a Euclidean space. When required, it will be generalized to the canonical duality pairing in the infinite dimensional setting.
\item $\succeq$ denotes an element-wise comparison between vectors $u, v \in \mathbb{R}^{p}$: $u \succeq v \iff u_{i} \geq v_{i}, \;\; \forall i \in \{1, 2, \dots ,p\}$. Other relational operators between vectors should also be understood element-wise.
\item The indicator function of a set $\mathbb{A}$ is:
\[
\chi_{\mathbb{A}}(x) \equiv
\left\{
\begin{aligned}
1  & \;\; \text{if} \;\; x \in \mathbb{A},  \\
0 & \;\; \text{otherwise.} \\
\end{aligned} \right.
\]
\item $\delta_{x_{0}}(x)$ denotes the distribution that is singular at $x=x_0$, being the Dirac delta-function w.r.t.\ Lebesgue measure in the case of continuous $x$. 
\item $\Delta_q$, \;$1 \leq q<\infty$, denotes the open probability simplex of dimension $q$. If $q>1$ and $x\in \Delta_q$, then the support of the conditional distribution, $\O{F}(x_{\backslash i} | x_i)$, is denoted by $(1-x_i)\Delta_{q-1}$, \;$0<x_i < 1$. 
\end{itemize}
The outline of the paper is as follows. In Section~\ref{HFPD-OT}, we state the mathematical problem and establish the duality result in the infinite dimensional case, hence deriving a formal characterization of the optimal Bayesian hyperprior \eqref{eq:hyperprior_init}. Section~\ref{parametric_hyperprior}  introduces the parametric hyperprior, and we provide a descriptive analysis in a low dimensional setting in Section~\ref{descriptive_analysis}. Meanwhile,  Section~\ref{Design_Kantorovitch_Potentials} proposes an algorithm for the computation of the optimal Kantorovitch potentials in this parametric setting.  In Section~\ref{simulations}, we apply the HFPD-OT formalism to a market matching problem in order to improve a contract diversity index, before closing the paper with our main conclusions in Section~\ref{conclusions}. 

\begin{figure}[h]  \centering
\begin{subfigure}{0.85\linewidth}
    \centering
        \includegraphics[width=0.5\linewidth]{ 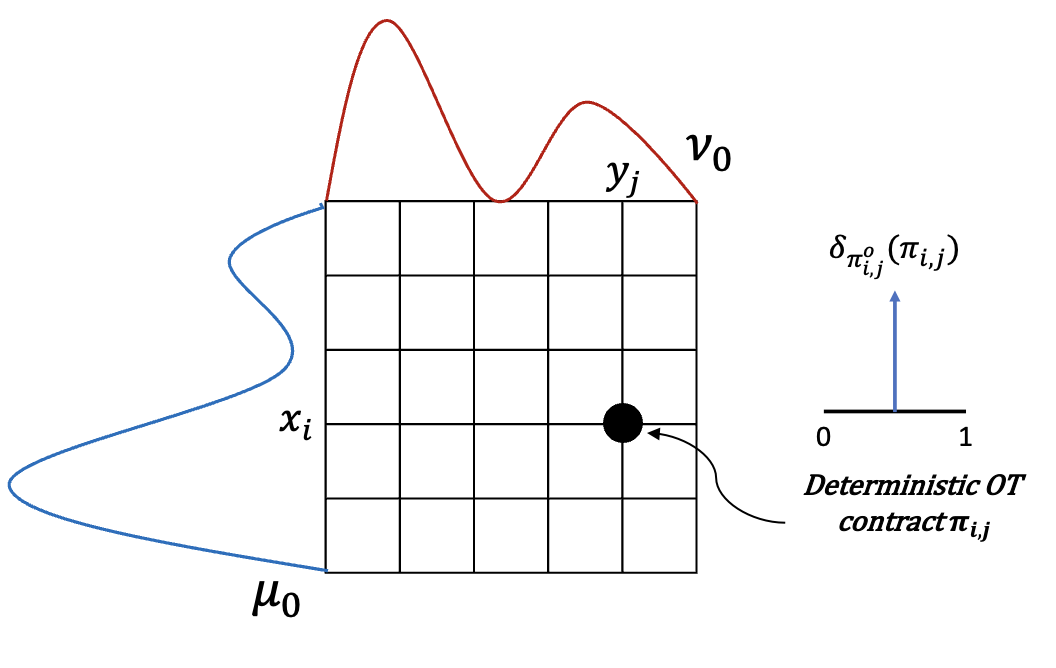}
        \caption{
        In the conventional base-level OT  setting, the transport plan, $\pi$, is deterministic, and so all the contracts, $\pi_{i,j}\in [0,1]$, are as well. Their respective (marginal) distributions are therefore singular at $\pi^{o}_{i,j}$, where $\pi^{o}$ denotes the OT plan \eqref{eq:OT_solution}.}
        \label{fig:fig21}
    \end{subfigure}
    \hspace{0.5cm}
   \begin{subfigure}{0.85\linewidth}
   \centering
        \includegraphics[width=0.5\linewidth]{ 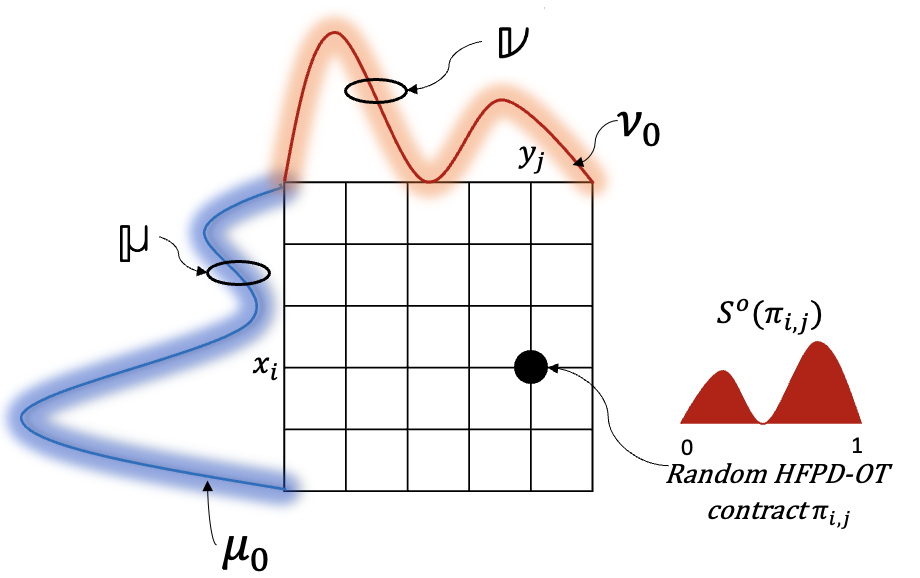}
        \caption{HFPD-OT acknowledges that uncertainty in $\mu$ and $\nu$ induces uncertainty in the transport plan, $\pi$, and therefore in the individual contracts, $\pi_{i,j}$. Hierarchical fully probabilistic design (HFPD)  endows $\pi$ with an optimal hyperprior, $\pi\sim\O{S}^o(\pi | K )$, whose marginals, $\mathsf{S}^{o}(\pi_{i,j})$, are the distributions of the contracts, $\pi_{i,j}\in [0, 1]$.}
        \label{fig:fig22}
    \end{subfigure}
   \begin{subfigure}{0.85\linewidth}
   \centering
        \includegraphics[width=0.5\linewidth]{ 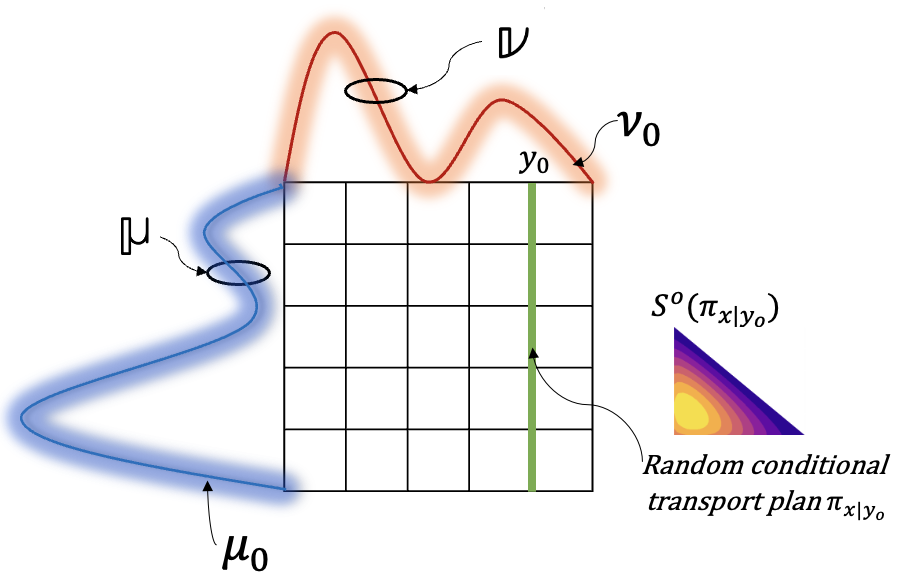}
        \caption{For a fixed $y_{0} \in \mathbb{\Omega}_{Y}$, HFPD-OT also acknowledges the conditional plans, $\pi_{x|y_{0}}$, as random processes, again equipped with their own optimal distributions, $\mathsf{S}^{o}(\pi_{x|y_{0}})$, consistent with $\pi\sim\O{S}^o(\pi | K )$.}
        \label{fig:fig23}
    \end{subfigure}
       \caption{Schematics which distinguish conventional base-level (i.e.\ deterministic) OT, in (a), from HFPD-OT, in (b) and (c).  For ease of illustration, we consider the finite dimensional specialization in Section~\ref{parametric_hyperprior}, but the ideas extend to the continuous setting. In HFPD-OT, uncertainty in the marginals, $\mu_0$ and $\nu_0$, induce uncertainty in    the joint ($\pi$) and conditional ($\pi_{x|y_{0}}$) transport plans, as well as in the individual contracts ($\pi_{i,j}$). All are optimally modeled in probability (i.e.\ they are random processes or variables, per the setting).  Here, a {\em contract}, $\pi_{i,j}\in [0, 1]$---see (a) and (b)---refers  to the normalized quantity  of resource (information,  assets, stock, {\em etc}.) transported  from agent $x_{i} \in \mathbb{\Omega}_{X}$ to agent $y_{j} \in \mathbb{\Omega}_{Y}$, in  delivering the (global) transport plan, $\pi$. }
       \label{fig:fig2}
\end{figure}

\section{Hierarchical Fully Probabilistic Design for (Optimal) Transport: HFPD-OT}\label{HFPD-OT}
The classical OT setting contemplates the transport plan as a purely deterministic object and frames the OT problem solely from an optimization perspective (Figure \ref{fig:fig21}). More precisely, FPD-OT \citep{sby2022fully}, which is a generalization of the classical EOT problem \citep{cuturi2013sinkhorn}, is built upon the following optimization problem:
\begin{equation}\label{eq:general_fpd_ot}
   \pi^{o}_{\O{OT}, \epsilon, \phi}(x,y | K) 
= \argmin_{\pi \in \mathbb{\Pi}(\mu_{0}, \nu_{0})}\mathsf{D}_{\mathsf{KL}}(\pi(x,y)||\pi_{\O{I}}(x,y | K)),
\end{equation}
where the base-level ideal design, $\pi_{\O{I}},$ with support in $\mathbb{\Omega}_{X} \times \mathbb{\Omega}_{Y}$, is defined as the following extended Gibbs kernel:
\begin{equation}\label{general_gibbs_kernel}
   \pi_{\O{I}}(x,y | K) \propto \exp\Bigl(\frac{-\mathsf{C}(x,y)}{\epsilon}\Bigr)\phi(x,y).
\end{equation}
$\mathsf{C}: \mathbb{\Omega}_{X} \times \mathbb{\Omega}_{Y} \rightarrow \mathbbm{R}^{+}$ denotes a continuous cost function, $\epsilon > 0$ is a smoothness (i.e.\ regularizing) parameter, and $\phi$ is a fixed distribution, which may be used to encode additional structural preferences in the design of the OT plan. $K$ (Section~\ref{conventions}) denotes the deterministic, domain-specific knowledge constraints, consisting of external or side-information gathered from the environment, and any other prior knowledge related to the problem being modeled. In the conventional base-level (i.e.\ deterministic) EOT setting, we impose these knowledge constraints in the form of deterministic marginal constraints $\mathbb{\Pi}(\mu_{0}, \nu_{0})$ \eqref{eq:Kset}. Importantly, when $\phi$ is instantiated as the uniform distribution, $\mathcal{U}$, with support in $\mathbb{\Omega}_{X} \times \mathbb{\Omega}_{Y}$, the resulting EOT solution converges in the $\Gamma$-sense to the Monge-Kantorovitch solution \citep{carlier2017convergence}: 
\begin{equation}\label{eq:OT_solution}
      \pi^{o}_{\O{OT}, \epsilon, \mathcal{U}}(x,y | K)  \xrightarrow[]{\epsilon \rightarrow 0} \pi^{o}_{\mathsf{OT}}(x,y | K) \equiv\argmin_{\pi \in \mathbb{\Pi}(\mu_{0}, \nu_{0})}\int_{\mathbb{\Omega}_{X}\times \mathbb{\Omega}_{Y}}\O{C(x,y)}\pi(x,y)d\lambda(x,y).
\end{equation}
In the sequel, we will denote the base-level OT solution simply by $\pi^{o}(x,y)$, and will not distinguish between EOT and OT solutions, unless required by the context. \\\\
In contrast to conventional, base-level OT---in which the transport plan, $\pi(x,y)$, is a deterministic object---HFPD-OT acknowledges that $\pi(x,y)$ is uncertain  (i.e.\ a random process), and  needs to be equipped with an appropriate hierarchical probability model (i.e.\ triple) (Figure \ref{fig:fig22}). Next, we deduce this optimal model, $\pi\sim\O{S}^o$ \eqref{eq:hyperprior_init}, using  the axiomatic Bayesian decision-making framework of hierarchical fully probabilistic design (HFPD).  

\subsection{The HFPD formulation of optimal transport}
Consider a probability model  in the hierarchical measurable space,
$(\mathbb{\Omega}_{\O{H}}, \mathcal{F}_{\mathbb{\Omega}_{\O{H}}})$, where $\mathbb{\Omega}_{\O{H}} \equiv\mathbb{\mathbb{\Omega}}_{X} \times \mathbb{\mathbb{\Omega}}_{Y} \times \mathbbm{P}(\mathbb{\mathbb{\Omega}}_{X} \times \mathbb{\mathbb{\Omega}}_{Y})$ and $\mathcal{F}_{\mathbb{\Omega}_{\O{H}}}$ is the $\sigma$-algebra of measurable sets in $\mathbb{\Omega}_{\O{H}}$. 
Then, $\pi(x,y) \in \mathbbm{P}(\mathbb{\Omega}_{X} \times \mathbb{\Omega}_{Y}) $ is a random process endowed with its own distribution, called the hyperprior, and denoted by $\mathsf{S}(\pi|K)$. The notation $\pi \sim \mathsf{S}(\pi|K)$ means that $\pi$ is distributed according to  a hyperprior, $\mathsf{S}(\pi|K)$, which is shaped by the  knowledge constraints, $K$ (specified below). Moreover, let $\mathcal{L}(\pi)$ denote the reference measure at the hierarchical level of the probability space. In the discrete case---when $\mathbbm{P}(\mathbb{\Omega}_{X} \times \mathbb{\Omega}_{Y})$ specializes to the probability simplex, $\Delta$---$\mathcal{L}(\pi)$ is instantiated as the Lebesgue measure, $\lambda (\pi )$. As in the conventional base-level OT setting, we assume that $\mathsf{S}(\pi|K)$ is absolutely continuous with respect to $\mathcal{L}(\pi)$, and we overload $\mathsf{S}(\pi|K)$ to denote its Radon-Nikodym derivative with respect to $\mathcal{L}(\pi)$. \\\\ 
Let $\mathbb{M}_{\O{H}}$ be the set of joint hierarchical Bayesian models with support in $\mathbb{\Omega}_{\O{H}}$. The joint hierarchical Bayesian model $\mathsf{M}(x, y, \pi|\O{S}, K) \in \mathbb{M}_{\O{H}}$---our new variational object---reads as follows:
\begin{equation}\label{eq:optimal_hier_model}
\begin{split}
\mathsf{M}(x, y, \pi|\O{S}, K) &= \O{M}(x,y|\pi,\O{S},K)\O{M}(\pi|\O{S}, K) \\
&= \pi(x,y|K)\O{S}(\pi|K) 
\end{split}
\end{equation}
\eqref{eq:optimal_hier_model} is a direct consequence of the conditional independence structure intrinsic to hierarchical modelling  (Figure  \ref{DAG_HFPD_OT}), and the fundamental definitions of $\pi$ and $\O{S}$. 
\\\\ 
\begin{definition}[Expected transport plan] \label{eq:expected_plan}
The random transport plan, $\pi \sim \O{S}(\pi|K)$  (\ref{eq:optimal_hier_model}), has the expected value,
\begin{equation}\label{eq:expected_plan}
\hat{\pi}_{\O{S}} (x,y | K) \equiv\O{E}_{\O{S}}[\pi]
\equiv\int_{\mathbb{P}(\mathbb{\Omega}_{X} \times \mathbb{\Omega}_{Y})} \pi(x,y|K)\O{S}(\pi|K)d\mathcal{L}(\pi).
\end{equation}
\end{definition}
Hence, the marginal model of $(x,y)$---and, therefore, the base-level transport plan induced by $\O{S}$---is $\hat{\pi}_{\O{S}}$ (\ref{eq:expected_plan}), as may be seen by integrating both sides of (\ref{eq:optimal_hier_model}) over $\pi \in \mathbbm{P}(\mathbb{\mathbb{\Omega}}_{X} \times \mathbb{\mathbb{\Omega}}_{Y})$:
\begin{equation}
    \label{eq:expplan2}
  \O{M}(x,y|K) = \hat{\pi}_{\O{S}}(x,y|K). 
\end{equation}
This is a necessary condition for consistent hierarchical Bayesian modelling, and arises because of the deterministic mapping, $(x,y)\rightarrow \pi(x,y)$, imposed by any realization of $\pi \sim \O{S}(\pi|K)$. 

From the foregoing, it is evident that the problem of hierarchical transport  model design is one of optimization of {\em deterministic\/} $\O{S} \in {\mathbb S}(\mathbbm{P}(\mathbb{\mathbb{\Omega}}_{X} \times \mathbb{\mathbb{\Omega}}_{Y}))$, noting that $\O{S}$ appears as a condition in (\ref{eq:optimal_hier_model}).  The challenge in designing the optimal hierarchical model over the set of transport plans in \eqref{eq:optimal_hier_model} is to optimally process the stochastic knowledge constraints imposed by the uncertain environment while being close to an ideal design $\O{M}_{\O{I}}$, which is used by the modeler to encode additional inductive biases and preferences in the HFPD-OT problem. \\ 
\begin{wrapfigure}{r}{0.3\textwidth} 
    \centering
    \includegraphics[width=0.2\textwidth,clip, trim = 0 20 0 0]{ 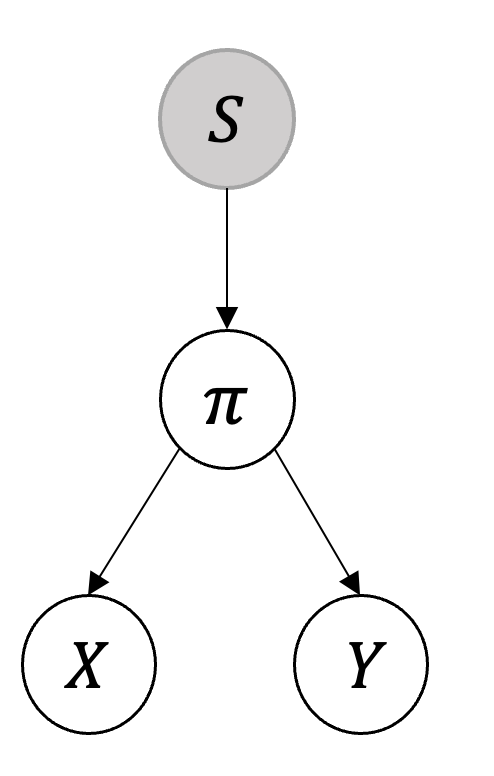}
    \caption{The conditional independence graph associated with HFPD-OT. Shaded nodes are observed. The arrows indicate the causal structure, where an arrow from one variable to a second indicates that the first variable causes the second.}
    \label{DAG_HFPD_OT}
\end{wrapfigure}

The generalized Bayesian inference framework considered here for the purpose of designing the optimal hierarchical model is Fully Probabilistic Design (FPD), introduced in \citep{KARNY2012105} and extended later to the hierarchical setting in \citep{QUINN2016532}.
Generalized Bayesian inference (GBI) is a set of techniques that extend the classical Bayesian inference method by updating the prior belief distribution using a loss function rather than the traditional likelihood function. Under incomplete model specification, the latter may indeed not exist \citep{Bissiri_2016}. However, FPD differs from other GBI techniques in two ways. First, FPD relies on the concept of ideal design in place of a prior, and allow the designer to elicit their personal preferences in the design process through an ideal, and usually unattainable, distribution $\O{M}_{\O{I}}(x,y,\pi | K) \in {\mathbb M}_{\O{H}}^c \equiv\mathbb{M} \smallsetminus \mathbb{M}_{\O{H}}$ (Figure \ref{fig:fig3}). More precisely, we assume that the ideal design factorizes as follows:
\begin{equation}\label{eq:joint_ideal_design}
\mathsf{M}_{\mathsf{I}}(x, y, \pi | K) \equiv\pi_{\mathsf{I}}(x,y | K)\mathsf{S}_{\mathsf{I}}(\pi |K)
\end{equation}
In other words, the joint ideal design, $\mathsf{M}_{\mathsf{I}}(x, y, \pi | K)$, is the base-level ideal design $\pi_{\mathsf{I}}$, modulated by the hierarchical ideal design $\mathsf{S}_{\mathsf{I}}$. Note that $\O{M}_{\O{I}}(x,y,\pi | K)$ is unattainable because $\pi_{\mathsf{I}}$ and $\mathsf{S}_{\mathsf{I}}$ are statistically independent models, and, as such, they may be conflicting in the following sense: \begin{equation}\label{eq:inconsistent_ideals}
\O{E}_{\mathsf{S}_{\O{I}}}(\pi) \neq \pi_{\O{I}}(x,y)
\end{equation}
This is reasonable when we recall that the ideal design is an entirely subjective object used to encode the designer's preferences (and representing their unattainable, zero-loss state of knowledge). By ranking the designer's preferences against this ideal design,  (hierarchical) FPD is consistent with  Savage's framework for Bayesian decision making \citep{LeonardSaravePersoProbabilities}. 
The consistent ranking of knowledge-constrained models \eqref{eq:optimal_hier_model} is via the KLD referenced to $\O{M}_{\O{I}}$. Hence,  the optimal hierarchical design, $\mathsf{M}^{o} (x,y,\pi|K)$, is formulated as follows:

\begin{equation}\label{eq:primal}
    (P): \;\;\;\;\; \mathsf{M}^{o} \in \argmin_{\mathsf{M} \in \mathbb{M}_{\O{H}}}\left\{\mathsf{D}_{\mathsf{KL}}\bigl(\mathsf{M}(x,y,\pi|K)|| \mathsf{M}_{\mathsf{I}}(x,y,\pi)\bigr)\right\},
\end{equation}
subject to:
\[\left\{
\begin{aligned}
\label{eq:primal_constraints}
    \O{E}_{\mathsf{S}}(\mathsf{D}_{\mathsf{KL}}(\mu||\mu_{0})) \leq \eta  \\
 \O{E}_{\mathsf{S}}(\mathsf{D}_{\mathsf{KL}}(\nu||\nu_{0})) \leq \zeta
\end{aligned} \right.\]
We note the following:
\begin{enumerate}
\item Since $\mathsf{D}_{\mathsf{KL}}\bigl(\cdot\;||\; \mathsf{M}_{\mathsf{I}}\bigr)$ is continuous, the space of joint hierarchical Bayesian distributions $\mathbb{M}_{\O{H}}$ is compact in the weak-* topology (see Appendix \ref{appendix}) and the constraint set is nonempty (we can for instance choose $\O{S} \equiv \delta_{\mu_{0} \otimes \nu_{0}}$), then the minimum is attained. 
\item Moreover, the optimal joint hierarchical model $\mathsf{M}^{o}$ is unique up to a set of measure 0. 
\end{enumerate}

The ideal design $\O{M}_{\O{I}}$ enters the KL divergence as the second fixed argument against which all feasible Bayesian hierarchical models are ranked. Importantly, note that the marginals in \eqref{eq:primal} are no longer modeled as deterministic, crisp objects. This assumption is now relaxed, allowing the modeler to express their uncertainty by viewing the marginals as random realizations of some underlying stochastic process. In particular, we describe this uncertainty in the form of moment constraints: the random marginals belong to uncertainty sets in the form of Kullback-Leibler balls, centered on $\mu_{0} \in \mathbbm{P}(\mathbb{\Omega}_{X})$ and $\nu_{0} \in \mathbbm{P}(\mathbb{\Omega}_{Y})$. The new knowledge-constrained set of consistent hierarchical Bayesian models---denoted by $\mathbb{M}_{K} \subseteq \mathbb{M}_{\O{H}}$---is augmented with the following linear moment constraints over the marginals:
\begin{equation} 
\begin{split}
\label{set:hierachical_model}
\mathbb{M}_{K} \equiv\Bigl\{\mathsf{M}(x,y,\pi|K) \;|\; \mathsf{M}(x,y,\pi|K) \in \mathbb{M}_{\O{H}} \;, \; \mu \in \bbmu \;\text{and}\; \nu \in \bbnu \Bigr\}
\end{split}
\end{equation} 
with the sets $\bbmu$ and $\bbnu$ defined as follows (Figure \ref{fig:fig22}): 
\begin{equation}\label{eq:moment_constraints1}
\bbmu \equiv\{\mu \in \mathbb{P}(\mathbb{\Omega}_{X})\; | \;\O{E}_{\mathsf{S}}\left[\mathsf{D}_{\mathsf{KL}}(\mu||\mu_{0})\right] \leq \eta\}
\end{equation}
\begin{equation}\label{eq:moment_constraints2}
\bbnu \equiv\{\nu \in \mathbb{P}(\mathbb{\Omega}_{Y}) \;| \;\O{E}_{\mathsf{S}}\left[\mathsf{D}_{\mathsf{KL}}(\nu||\nu_{0})\right] \leq \zeta \}
\end{equation}
where $\eta \geq 0$ and $\zeta \geq 0$ are prior-elicited KL radii, that express the degree of uncertainty the designer is placing over the marginals. \\\\
As we will see in the sequel, the interaction between the base-level and hierarchical ideals, on one hand, and the knowledge constraints on the other, is what gives rise to the Gibbsian form of the hyperprior in \eqref{eq:hyperprior_init}. \\\\ We  now state the main result of the paper. 

\vspace*{.1cm}

\begin{theorem}\label{main_theorem} 
Let (P) be the HFPD-OT Primal problem, defined in \eqref{eq:primal}. 
\begin{enumerate}
\item (P) is equivalent to the following optimization problem over the set of hierarchical Bayesian models $\mathbb{M}_{\O{H}}$ \eqref{set:hierachical_model}:
\begin{equation}\label{eq:new_primal}
    (P): \;\;\;\;\; \O{M}^{o}(x,y,\pi) \in \argmin_{\O{M} \in \mathbb{M}_{\O{H}}} \left\{\mathsf{D}_{\mathsf{KL}}\bigl(\mathsf{M}(x,y,\pi|K)|| \hat{\pi}_{\tilde{\O{S}}}(x,y)\tilde{\mathsf{S}}(\pi|K)\bigr) \right\},
\end{equation}
subject to
\[\left\{
\begin{aligned}
    \O{E}_{\mathsf{S}}(\mathsf{D}_{\mathsf{KL}}(\mu||\mu_{0})) \leq \eta  \\
 \O{E}_{\mathsf{S}}(\mathsf{D}_{\mathsf{KL}}(\nu||\nu_{0})) \leq \zeta
\end{aligned} \right.\]
where
\begin{equation}\label{eq:new_ideal_design}
\tilde{\mathsf{S}}(\pi|K) \equiv\mathsf{S}_{\mathsf{I}}(\pi)\exp\Bigl(-\mathsf{D}_{\mathsf{KL}}\bigl(\pi(x,y)||\pi_{\mathsf{I}}(x,y)\bigr)\Bigr).
\end{equation}
\item The optimal hyperprior $\mathsf{S}^o (\pi | K)$ reads as follows:
\begin{equation}  \label{eq:hyperprior}
\mathsf{S}^o (\pi | K) \propto
 \exp\left(-\lambda_1^o {\mathsf{D}_{\mathsf{KL}}}(\mu || \mu_0)\right)  \tilde{\mathsf{S}}(\pi|K)
\exp\left(-\lambda_2^o {\mathsf{D}_{\mathsf{KL}}}(\nu || \nu_0)\right), \;\; \mathcal{L}\text{-a.e.}
\end{equation}
\item The Dual program associated with the primal $(P)$ \eqref{eq:new_primal} reads
\begin{equation}\label{eq:dual}
\begin{split}
    (D): \;\;\;\; \sup_{\boldsymbol{\lambda} \succeq 0} \left\{ \log\left(\mathsf{N}(\boldsymbol{\lambda})\right) - \boldsymbol{\lambda}^\intercal \boldsymbol{\theta} \right\},
\end{split}
\end{equation}
where \begin{equation}\label{normalization}
\mathsf{N}(\boldsymbol{\lambda}) \equiv\left(\int_{\mathbbm{P}(\mathbb{\Omega}_{X} \times \mathbb{\Omega}_{Y})} \tilde{\O{S}}(\pi|K) \exp\left(- <\boldsymbol{\lambda}, \mathsf{R}(\pi)>-1\right)d\mathcal{L}(\pi) \right)^{-1},
\end{equation}
and $\boldsymbol{\lambda}$ $\equiv$ $\begin{bmatrix} 
\lambda_{1} \\
 \lambda_{2}
\end{bmatrix}$, $\boldsymbol{\theta} \equiv\begin{bmatrix} \eta \\ \zeta \end{bmatrix}$,   $\mathsf{R}(\pi)$ $\equiv$ $\begin{bmatrix} \label{moment_constraints} \mathsf{D}_{\mathsf{KL}}(\mu||\mu_{0}) \\ \mathsf{D}_{\mathsf{KL}}(\nu||\nu_{0})
\end{bmatrix}$. \\\\
Moreover, strong duality holds, i.e.\ the optimal Kantorovitch potentials, $\boldsymbol{\lambda}^{o}$ in \eqref{eq:hyperprior}, are the solution of the dual problem \eqref{eq:dual},
\begin{equation}
    \label{eq:optKant}
    \boldsymbol{\lambda}^o \equiv\argmax_{\boldsymbol{\lambda} \succeq 0} \left\{ \log\left(\mathsf{N}(\boldsymbol{\lambda})\right) - \boldsymbol{\lambda}^\intercal \boldsymbol{\theta} \right\},
\end{equation}  
and the maximum of the dual problem is attained: $\min_{\O{M}}(P) = \max_{\boldsymbol{\lambda}}(D)$.
\end{enumerate}
\end{theorem}

\begin{proof_method}
Results (1) and (2) of the Theorem can be proved using basic algebraic manipulations. However, we opt here for a derivation based on information processing arguments, so as to gain more intuition about the design of the hyperprior in the hierarchical setting. \\\\
Given the factorized joint ideal design in \eqref{eq:joint_ideal_design}, the optimal hyperprior $\mathsf{S}^{o}(\pi|K)$ emerges via two sequential knowledge-processing steps (Figure \ref{fig:fig3}), addressed in the first two of the following items:
\begin{enumerate}
\item Adapting the ideal design and processing the hyperprior without knowledge constraints $K$. The purpose of this first step is to guide the optimization problem $(P)$ in \eqref{eq:primal} from a possibly inconsistent ideal, $\mathsf{M}_{\O{I}}$ \eqref{eq:inconsistent_ideals}, to a new consistent target (step 1 in Figure \ref{fig:fig3}). The adapted hyperprior, $\tilde{\mathsf{S}}$ (\ref{eq:new_ideal_design}), expresses the best compromise between possibly conflicting ideals. It involves the Gibbs-type modulation of the hierarchical ideal design $\mathsf{S}_{\mathsf{I}}$ via a term  that depends on the base-level ideal design $\pi_{\mathsf{I}}$ (Theorem 1 in \citep{QUINN2016532}). The optimal hierarchical model $\tilde{\mathsf{M}} \in \mathbb{M}_{\O{H}}$ is a boundary point in the convex set $\mathbb{M}_{\O{H}}$ and is inferred from \eqref{eq:optimal_hier_model} as follows:
\begin{equation}
\label{eq:Mtilde}
\tilde{\mathsf{M}}(x,y, \pi|K) = \hat{\pi}_{\tilde{\mathsf{S}}}(x,y|K) \tilde{\mathsf{S}}(\pi|K)
\end{equation}
where $\hat{\pi}_{\tilde{\mathsf{S}}}$ is the expected transport plan \textit{w.r.t} $\tilde{\mathsf{S}}$ and follows from \eqref{eq:expected_plan}.
\item Processing the two marginal constraints specified in the knowledge set $K$. This step leads to the new optimization problem stated in \eqref{eq:new_primal}, which results in the optimal hyperprior \eqref{eq:hyperprior} (Theorem 3 in \citep{QUINN2016532}). Each of the marginal constraints induces a MaxEnt Gibbs term that modulates the hyperprior obtained in Step 1.
And the resulting optimal hierarchical model $\mathsf{M}^{o} \in \mathbb{M}_{K} \subseteq \mathbb{M}_{\mathsf{H}}$---which is also a boundary point in the convex set $\mathbb{M}_{K}$---reads as follows:
\begin{equation}
\label{eq:Mo}
\mathsf{M}^{o}(x,y, \pi|K) = \hat{\pi}_{\mathsf{S}^{o}}(x,y|K)\mathsf{S}^{o}(\pi|K) 
\end{equation}
where $\hat{\pi}_{\mathsf{S}^{o}}$ follows similarly from \eqref{eq:expected_plan}.
\item It remains to prove the strong duality result and formally characterize the Kantorovitch potentials in \eqref{eq:dual}. The details of this proof are provided in Appendix~\ref{appendix}. There, we prove strong duality in the infinite dimensional case by relying on the classical Fenchel-Rockafellar duality theorem \citep{pjm_1102992608}, \citep{Villani2008OptimalTO}. More precisely, we demonstrate that the conditions required by the theorem are satisfied in the hierarchical Bayesian setting of HFPD-OT, and we derive the dual problem $(D)$. 
\end{enumerate} 
\end{proof_method}
\begin{figure}[h!]  \centering
\includegraphics[width=.64\linewidth]{ 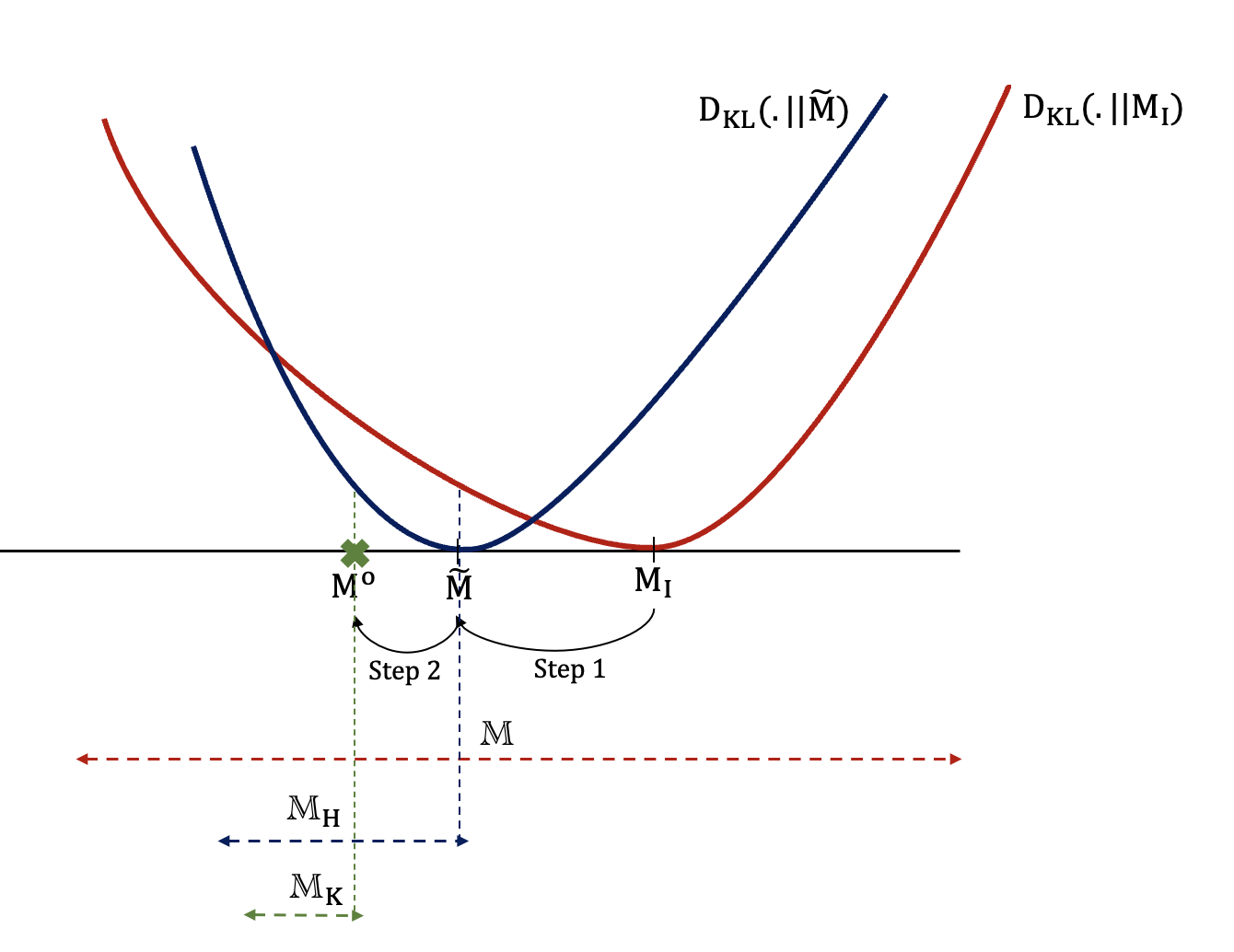}
\caption{ A sequential information-processing view of the optimal hierarchical model, $\mathsf{M}^o \equiv\hat{\pi}_{{\mathsf S}^o} \mathsf{S}^o$, used in the proof method (\ref{eq:Mo}). {\em First}, the inductive biases expressed via the  hierarchical ideal model, $\mathsf{M}_\O{I} \in {\mathbb M}_{\O{H}}^c$ (\ref{eq:joint_ideal_design}), are processed to yield a new optimization problem over a constrained set $\mathbb{M}_{\mathsf{H}}$, whose solution, $\tilde{\mathsf M}$  (\ref{eq:Mtilde}), is on the boundary of  $\mathbb{M}_{\mathsf{H}}$. {\em Second}, the knowledge constraints,  $K$, are processed, further reducing  the feasible set to the subset, $\mathbb{M}_{K}$ (\ref{set:hierachical_model}). The optimal hierarchical model  is $\mathsf{M}^{o}$ (on the boundary of $\mathbb{M}_K$), s.t.\ $\pi \sim \mathsf{S}^o (\pi | K)$ (\ref{eq:hyperprior}).
}
    \label{fig:fig3} 
\end{figure}
By sampling random realizations from our optimal hyperprior, we can design randomized and diverse transport policies in lieu of an immutable and fixed OT plan. This randomization principle is depicted in Figure $\ref{fig:HFPD_OT}$. More precisely, the design of the optimal hyperprior over the space of transport plans is a twofold process:
\begin{enumerate}
\item The knowledge constraints $K$ are processed to yield the optimal hyperprior $\eqref{eq:hyperprior}$. This mainly requires conditioning the Kantorovitch potentials on the uncertainty radii, $(\eta, \zeta)$ (Figure \ref{fig:HFPD_OT_FIRST_STEP}).
\item Once the optimal hyperprior is available, random transport strategies are sampled and used in subsequent transport problems, in lieu of a crisp OT plan. Importantly, having access to a generative model over the space of transport plans provides us with the statistical devices to assess and reason about the intrinsic uncertainty in the transport problem (Figure \ref{fig:HFPD_OT_SEC_STEP}). The expected transport $\hat{\pi}_{\mathsf{S}^{o}}$ plan is obtained from \eqref{eq:expected_plan}.
\end{enumerate}

\begin{remark}\label{uncertainty_potentials}
The Kantorovitch potentials $\lambda_{1} = \lambda_{1}(\eta, \zeta)$ and $\lambda_{2} = \lambda_{2}(\eta, \zeta)$ express the degree of uncertainty in the input data---i.e. the marginals. Depending on their values, they give rise to two interesting extremal modalities, that vary from high uncertainty to perfect characterization of the marginals:
\begin{itemize}
\item If $\eta \rightarrow \infty$ and $\zeta \rightarrow \infty$, it is straightforward from \eqref{eq:dual} that the solution of the dual is achieved when $\boldsymbol{\lambda}^{o}=\boldsymbol{0}$. This is also a direct consequence of complementary slackness. It follows that
\begin{equation}\label{eq:unspecified_marginals}
    \mathsf{S}^{o}(\pi|K)  \xrightarrow[]{\eta \to \infty, \;\zeta \to \infty} \tilde{\mathsf{S}}(\pi|K). 
    \end{equation}
    In other words, when the uncertainty around the marginals is unbounded, the optimal hyperprior is mainly characterized---see (\ref{eq:new_ideal_design})---by  the hierarchical ideal design modulated by  a Gibbsian term that depends on $\pi_{\mathsf{I}}$.
    \item If $\eta \rightarrow 0$ and $\zeta \rightarrow 0$, the uncertainty in the marginals vanishes and \textit{learning}\footnote{In the context of (H)FPD, learning (i.e.\ inductive inference) refers to the optimal processing of  knowledge constraints into the  hyperprior: $K \rightarrow \O{S}^o (\pi | K)$. For more discussion on the role of FPD in furnishing generalized settings of Bayes' rule, see \citep{kracik2005merging}.} is maximal, leading to $\mu \rightarrow \mu_{0}$ and $\nu \rightarrow \nu_{0}$, or equivalently $\pi \rightarrow \tilde{\pi} \in \mathbb{\Pi}(\mu_{0}, \nu_{0})$. It follows from the dual \eqref{eq:dual} that the maximum is attained when $\boldsymbol{\lambda}^{o} \rightarrow \infty$, and we achieve the limit, 
    \begin{equation}\label{eq:degenerate_distribution}
    \mathsf{S}^{o}(\pi|K) \xrightarrow[]{\eta \to 0, \;\zeta \to 0} \tilde{\mathsf{S}}(\pi|K) 
\chi_{\mathbb{\Pi}(\mu_{0}, \nu_{0})}(\pi).
    \end{equation}
    In other words, the hyperprior concentrates on the OT manifold,  $\mathbb{\Pi}(\mu_{0}, \nu_{0})$ (\ref{eq:Kset}). 
    This concentration behaviour is reminiscent of the Laplace-Bernstein-Von Mises convergence theorem \citep{bernstein-von-mises}.

   \end{itemize}
\end{remark}
\begin{remark}\textbf{Conventional Base-level OT}
 Consider further the regime of perfect specification of the marginals, i.e. $\eta \rightarrow 0$, $\zeta \rightarrow 0$ (Remark~\ref{uncertainty_potentials}). The conjugate choice of the ideal hyperprior, $\mathsf{S}_{\O{I}}$, has the following Gibbs form:
    \begin{equation}\label{eq:hierachical_ideal_design}
\mathsf{S}_{\O{I}}(\pi|K) \propto \exp(- \alpha \mathsf{D}_{\mathsf{KL}}(\pi(x, y)||\pi_{\O{I}}(x,y))).
    \end{equation}
Here,  $\alpha > 0$ plays the role of the inverse-temperature. Substituting (\ref{eq:hierachical_ideal_design})  into \eqref{eq:degenerate_distribution}, the optimal hyperprior becomes 
\begin{equation}\label{eq:hyperprior_base_eot}
    \mathsf{S}^{o}(\pi|K) \xrightarrow[]{\eta \to 0, \;\zeta \to 0} \exp(-(\alpha + 1) \O{D}_{\O{KL}}(\pi||\pi_{\O{I}}))
\chi_{\mathbb{\Pi}(\mu_{0}, \nu_{0})}(\pi).
    \end{equation}
When $\pi_{\O{I}}$ is the extended Gibbs kernel (\ref{general_gibbs_kernel})---where we instantiate $\phi$ as the uniform distribution with support in $\mathbb{\Omega}_{X} \times \mathbb{\Omega}_{Y}$---the minimum of $\O{D}_{\O{KL}}(\pi||\pi_{\O{I}})$ in \eqref{eq:hyperprior_base_eot} is exactly achieved at the EOT solution \eqref{eq:OT_solution}:
        \begin{equation}
 \pi^{o}(x,y | K) = \argmin_{\pi \in \mathbb{\Pi}(\mu_{0}, \nu_{0})} \mathsf{D}_{\mathsf{KL}}(\pi(x, y)||\pi_{\O{I}}(x,y)).
    \end{equation}
The latter can be recovered when $\alpha \rightarrow \infty$, for example by simulated annealing \citep{Delahaye2019}:  
\begin{equation}\label{eq:base_level_eot}
\O{S}^{o}(\pi|K) \xrightarrow[]{\eta \to 0, \;\zeta \to 0, \alpha \to \infty} \delta_{\pi^{o}}(\pi).
\end{equation}
\end{remark}


\begin{figure}[h]  \centering
    \begin{subfigure}{0.8\linewidth}
    \centering
        \includegraphics[width=0.6\linewidth]{ 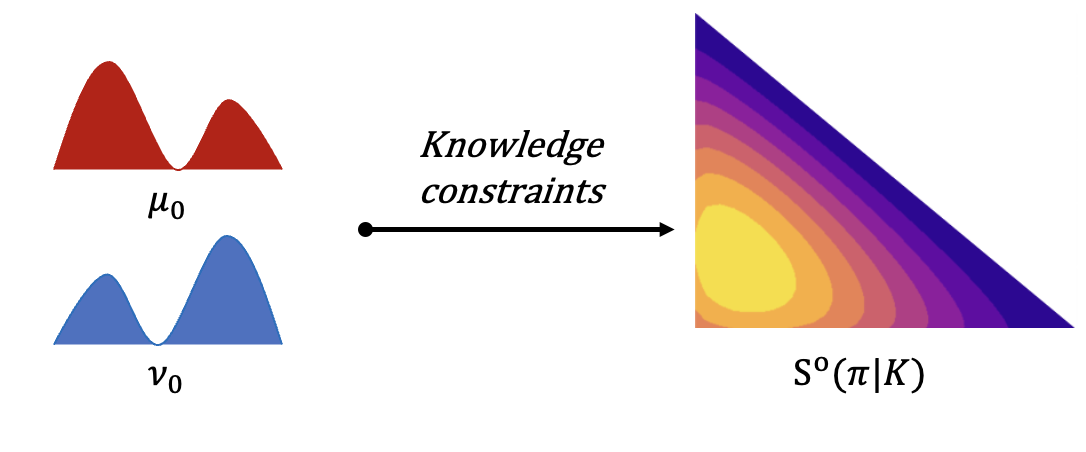}
        \caption{First, the optimal hyperprior, $\mathsf{S}^{o}(\pi|K)$, is computed, by processing the marginal knowledge constraints into the optimal Kantorovitch potentials  (\ref{eq:optKant}).}
        \label{fig:HFPD_OT_FIRST_STEP}
    \end{subfigure}
    \hspace{0.5cm}
   \begin{subfigure}{0.8\linewidth}
   \centering
        \includegraphics[width=0.55\linewidth]{ 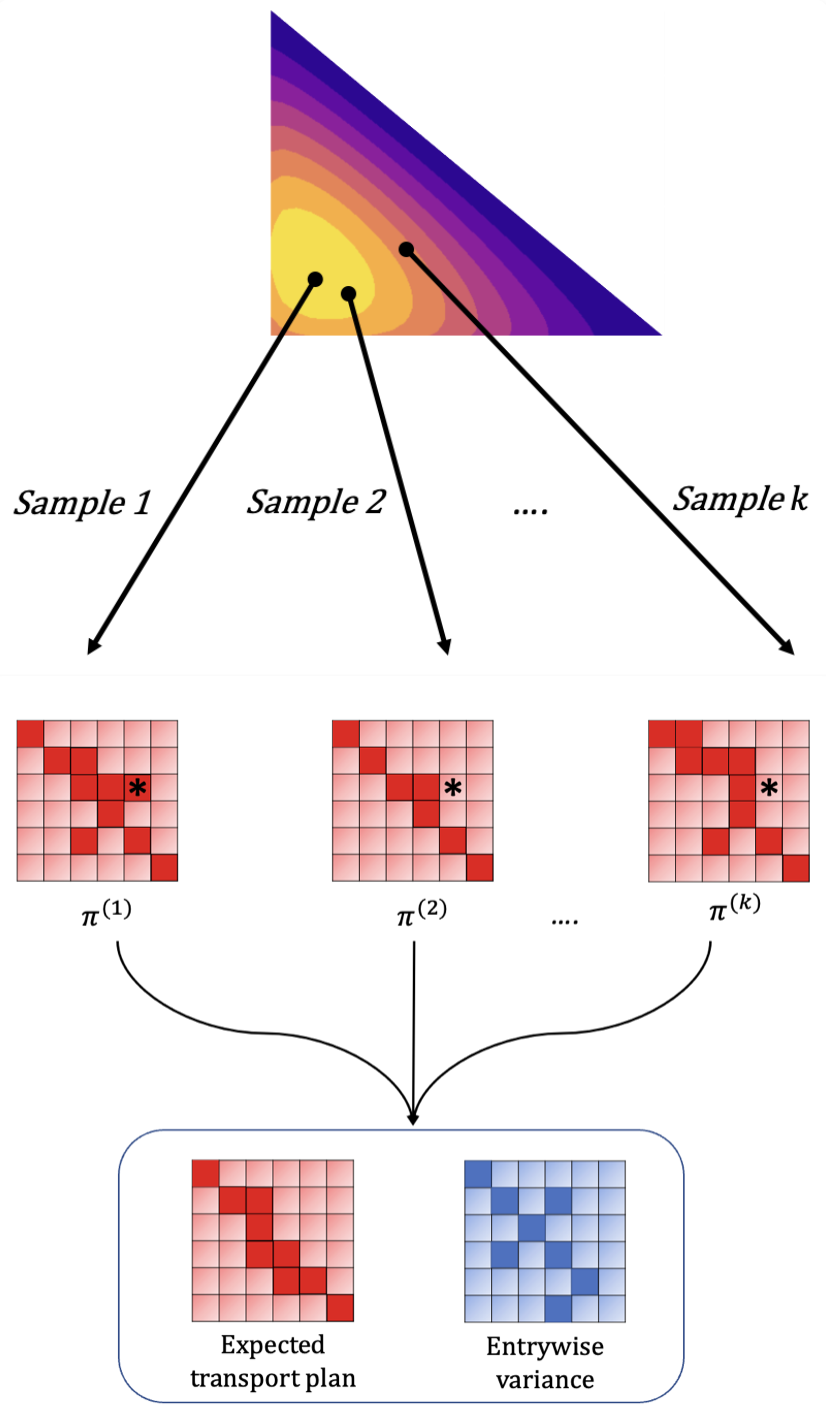}
        \caption{Once elicited, the optimal hyperprior, $\mathsf{S}^{o}(\pi|K)$, can be used to sample random transport plans, $\pi^{(k)}$. These random samples of plans can be used in two important inference steps: \begin{enumerate*}\item the expected transport plan (bottom left), $\hat{\pi}_{\O{S}^o} (x,y | K)$ (\ref{eq:expected_plan}), can be used in downstream transport tasks, in lieu of the conventional base-level OT plan, $\pi^o$ (\ref{eq:general_fpd_ot});  and \item measures of uncertainty (bottom right) in the form of entry-wise (i.e.\ contract) variances, or other summary statistics (including higher-order correlation structure between contracts) can be designed to inform the decision-making process. The asterisk (*) highlights an example of a contract that experiences diversified transport policies, enabled by randomized HFPD-OT. \end{enumerate*}}
        \label{fig:HFPD_OT_SEC_STEP}
    \end{subfigure}

\caption{The two-step principle underlying HFPD-OT. $\mathsf{S}^{o}(\pi|K)$ is a generative model (i.e.\ a distribution) of random transport plans, $\pi$. Realizations, $\pi^{(k)}$, of $\pi$ can be sampled from $\mathsf{S}^{o}(\pi|K)$, and these samples can then be used to estimate an expected transport plan (\ref{eq:expected_plan}) for downstream transport problems, via ergodic averaging. In addition, HFPD-OT enables a principled analysis of the intrinsic uncertainty in the transport problem.}
\label{fig:HFPD_OT}
\end{figure}

\section{The  HFPD-OT hyperprior in the parametric case}\label{parametric_hyperprior}
As already noted, no special assumptions have been made in respect of the hierarchical transport model (\ref{eq:optimal_hier_model}), and so   (\ref{eq:hyperprior}) is the HFPD-OT hyperprior for the nonparametric (transport) process, $\pi \in \mathbbm{P}(\mathbb{\Omega}_{X} \times \mathbb{\Omega}_{Y})$. The finite case---i.e.\ $\#(\mathbb{\Omega}_{X} \times \mathbb{\Omega}_{Y}) < \infty$---induces the parametric setting of HFPD-OT, with $\pi$ defined in the usual way \textit{w.r.t.} the counting measure, and $\O{S}^o (\pi | K)$ defined on a $K$-constrained subset (\ref{eq:primal_constraints})  of the simplex. This allows us to easily visualize key properties of $\O{S}^o (\pi | K )$ in a low dimensional setting, and, importantly, to develop algorithms for computing random draws (Figure \ref{fig:HFPD_OT_SEC_STEP}), $\pi^{(k)} \sim \O{S}^o (\pi | K )$, from the HFPD-OT parametric hyperprior (\ref{eq:hyperprior}),
via approximation of the Kantorovitch potentials (\ref{eq:optKant}).   
 
\subsection{Descriptive analysis of the parametric HFPD-OT hyperprior, $\O{S}^o (\pi | K ) $}\label{descriptive_analysis}
In the finite, parametric case---which we will pursue in the rest of this paper---$x\in\mathbb{\Omega}_{X}\equiv\{x_1, \ldots, x_i, \ldots, x_m\}$ and 
$y\in\mathbb{\Omega}_{Y}\equiv\{y_1, \ldots, y_j, \ldots, y_n\}$, with $2\leq m < \infty$ and $2\leq n < \infty$. 
We refer to $\mathbb{\Omega}_{X}$ and $\mathbb{\Omega}_{Y}$ as the sets of {\em source agents\/} and {\em target agents}, respectively. Then, the base-level distributions are uncertain multinomials, with densities $\mu = \sum_{i=1}^{m}\mu_{i}\delta_{x_{i}}$, $\nu = \sum_{j=1}^{n}\nu_{j}\delta_{y_{j}}$ and $\pi = \sum_{j=1}^{n} \sum_{i=1}^{m}\pi_{i,j}\delta_{x_i,y_j}$. The associated pmfs are structured as vector-matrix objects, and also denoted by the same symbols: $\mu \in \Delta_{m-1}$, $\nu \in \Delta_{n-1}$ and $\pi \in \Delta_{mn-1}$.
 Without loss of generality, we consider the following class of conjugate\footnote{We consider a weak form of conjugacy~\citep{10.4108/icst.valuetools.2011.246122}, where the processing of the ideal design, $\O{S}_{\O{I}} (\pi | K)$,  via hierarchical FPD yields an optimal hyperprior, $\O{S}^o (\pi | K ) $, of the {\em same\/} functional form.} hierarchical ideal designs, parameterized by fixed $ {\boldsymbol \lambda}_{\O{I}} \succeq 0$ (we absorb the parameter conditions---here, ${\boldsymbol \lambda}_{\O{I}}$, $\mu_0$ and $\nu_0$---into the Jeffreys' notation, $K$):
\begin{equation}
\mathsf{S}_{\O{I}}(\pi|K) \propto \prod_{i=1}^{m}\Bigl(\frac{\mu_{i}}{\mu_{0,i}}\Bigr)^{-\lambda_{\O{I}, 1}\mu_{i}}\prod_{j=1}^{n}\Bigl(\frac{\nu_{j}}{\nu_{0,j}}\Bigr)^{-\lambda_{\O{I}, 2}\nu_{j}}
\end{equation}
The base-level ideal design, $\pi_\O{I} (x,y | K)$, has the form of the extended Gibbs kernel \eqref{general_gibbs_kernel}, consistent with the FPD-OT setting. We further specialize $\phi (x,y)$ to the uniform case, $\phi(\cdot) \equiv{\mathcal U}$, yielding the following form of the parametric hyperprior:

\vspace*{.1cm}

\begin{definition}[HFPD-OT hyperprior for the parametric transport plan]\label{hyperprior_definition}
The transport hyperprior \eqref{eq:hyperprior} in the case of a domain, $\mathbb{\Omega}_X \times \mathbb{\Omega}_Y$, of finite cardinality, $m\times n$, is parametric, with parameters $(\lambda_{1}^{o}, \lambda_{2}^{o}, \lambda_{\O{I}, 1}, \lambda_{\O{I}, 2}, \mu_{0}, \nu_{0}, \pi_{\O{I}})$, and support on the probability simplex $\Delta_{m\times n -1}$. It is absolutely continuous \textit{w.r.t.}\  Lebesgue measure, $\lambda$, with density 
\begin{equation}\label{eq:parametric_hyperprior}
\mathsf{S}^{o}(\pi|K) \propto \prod_{i=1}^{m}\prod_{j=1}^{n}\Bigl(\frac{\mu_{i}}{\mu_{0,i}}\Bigr)^{-(\lambda_{\O{I}, 1}+\lambda_{1}^{o})\mu_{i}}\Bigl(\frac{\nu_{j}}{\nu_{0,j}}\Bigr)^{-(\lambda_{\O{I}, 2}+\lambda_{2}^{o})\nu_{j}}\Bigl(\frac{\pi_{i,j}}{ 
\pi_{\O{I}, i,j}}\Bigr)^{-\pi_{i,j}}\;\; \lambda\text{-a.e.},
\end{equation}
with the ideal design having the following Gibbs form: 
\[
\pi_{\O{I}, i,j} \propto \exp\Bigl(-\frac{\mathsf{C}(x_i, y_j)}{\epsilon}\Bigr)
\]
\end{definition}
The number of prior parameters, encoding $K$ in  \eqref{eq:parametric_hyperprior}, is $ (m+1)\times(n+1)$. 
 This endows the HFPD-OT hyperprior design with far more expressivity (i.e.\ degrees-of-freedom (dofs)) than default distributions on the probability simplex. For instance, a Dirichlet distribution of $\pi$ in this finite setting has $m+n+1$ fewer dofs.  
 
 \begin{remark} [Inference with the HFPD-OT hyperprior, $\O{S}^o (\pi | K ) $]
 The {\em normalizing constant\/} of the HFPD-OT hyperprior (\ref{eq:parametric_hyperprior})  is not available in closed form. A full study of its numerical approximation will be the subject of future work. 
\\\\
 The {\em marginal distribution\/} of  $\pi_{1:k, 1:l}\in \Delta_{k\times l-1}$,  being the sub-matrix of $\pi$ associated with the contracts, $\pi_{ij}$, $1\leq i \leq k < m$ and $1\leq j \leq  l < n$, is
\begin{equation}\label{eq:block_marginal}
\O{S}^o(\pi_{1:k, 1:l}|K) = \hspace*{-.8cm}
\int\displaylimits_{(1-w_{kl})\Delta_{mn-kl-1}}
\hspace*{-.9cm}\O{S}^o(\pi|K)
d\pi_{\backslash (1:k, 1:l)},
\end{equation}
where $w_{kl} \equiv\sum_{j=1}^l \sum_{i=1}^k \pi_{ij}$, and $\pi_{\backslash (1:k, 1:l)}$ denotes the complement of $\pi_{1:k, 1:l}$ in $\pi$.
In particular, the marginal distribution of $\pi_{k,l}\in(0,1)$---i.e.\ of the $(k,l)$th random contract, being the normalized mass (probability) transported from the  $k$th source node and the $l$th target node---is
\begin{equation}\label{eq:marginal}
\O{S}^o(\pi_{kl}|K) = \hspace*{-.8cm}
\int\displaylimits_{(1-\pi_{kl})\Delta_{mn-2}}
\hspace*{-.7cm}\O{S}^o(\pi|K)
d\pi_{\backslash (k,l)}.
\end{equation}
\\\\ 
Finally,  the HFPD-optimal {\em full conditional distribution\/} of the $(k,l)$th contract---having fixed all the others at specific probabilities, $\pi_0{_{\backslash (k,l)}}$---is  
\begin{equation}\label{EE_cond}
\mathsf{S}^{o}(\pi_{k,l}|\pi_0{_{\backslash (k,l)}}, K) \propto \mathsf{S}^{o}(\pi_{k,l},\pi_0{_{\backslash (k,l)}} |K) \;\chi_{(0,1-c_{kl})}(\pi_{kl}),
\end{equation}
where 
$c_{kl} \equiv  \hspace*{-.4cm}
\underbrace{\sum_{j=1}^n \sum_{i=1}^m}_{(i,j)\notin \{(k,l), (m,n)\}}\hspace*{-.6cm} \pi_0{_{i,j}}$.
\vspace*{.03cm}
\end{remark}
\subsubsection{Illustration in the $m = n = 2$ case}
To gain further insight into the parametric HFPD-OT hyperprior, $\O{S}^o (\pi | K)$ \eqref{eq:parametric_hyperprior}, we explore its location and shape in the $m = n = 2$ case. Then,  $\mathsf{S}^o(\pi_{11}, \pi_{12}, \pi_{21}|K)$ has support in the three-dimensional simplex, i.e.\ $(\pi_{11}, \pi_{12}, \pi_{21}) \in \mathbbm{P}(\mathbb{\Omega}_{X} \times \mathbb{\Omega}_{Y})\equiv \Delta_3$  
 (Figure \ref{fig:hyperprior_low_dim}). 
  We assume that $\boldsymbol{\lambda}^o\gg \boldsymbol{\lambda}_{\mathsf I}$, which corresponds to the knowledge-dominated regime~\citep{Jeffreys1939-JEFTOP-5} in which the ideal in (\ref{eq:new_ideal_design}) is diffuse in comparison with the $K$-dependent modulating terms in (\ref{eq:hyperprior}).
 In this case, (\ref{eq:parametric_hyperprior}) specializes to: 
\begin{equation}\label{hyperprior}
\begin{split}
\mathsf{S}^{o}(\pi_{11}, \pi_{12}, \pi_{21}|K) \propto  \Bigl(\frac{\pi_{11}+\pi_{12}}{\mu_{0,1}}\Bigr)^{-\lambda_{1}^{o}(\pi_{11}+\pi_{12})} \Bigl(\frac{1-\pi_{11}-\pi_{12}}{1-\mu_{0,1}}\Bigr)^{-\lambda_{1}^{o}(1-\pi_{11}-\pi_{12})}\Bigl(\frac{\pi_{11} + \pi_{21}}{\nu_{0,1}}\Bigr)^{-\lambda_{2}^{o}(\pi_{11} + \pi_{21})}\hspace*{-1.7cm}\times & \\ 
\Bigl(\frac{1-\pi_{11}- \pi_{21}}{1-\nu_{0,1}}\Bigr)^{-\lambda_{2}^{o}(1-\pi_{11} -\pi_{21})}  \Bigl(\frac{\pi_{11}}{\pi_{\O{I}, 11}}\Bigr)^{-\pi_{11}} \Bigl(\frac{\pi_{12}}{\pi_{\O{I}, 12}}\Bigr)^{-\pi_{12}} \Bigl(\frac{\pi_{21}}{\pi_{\O{I}, 21}}\Bigr)^{-\pi_{21}} 
\Bigl(\frac{1-\pi_{11}-\pi_{12}-\pi_{21}}{1-\pi_{\O{I}, 11}-\pi_{\O{I}, 12}-\pi_{\O{I}, 21}}\Bigr)^{-(1-\pi_{11}-\pi_{12}-\pi_{21})} \hspace*{-2.4cm}
\end{split}
\end{equation}
 \begin{wrapfigure}{r}{0.22\textwidth} 
    \centering
    \includegraphics[width=0.18\textwidth]{ 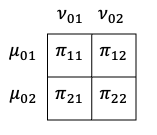}
   \caption{Schematic of an uncertain transport plan in the $\Delta_{3}$ simplex, annotating  the corresponding nominal (i.e.\ prior-specified) row and column marginals. The $(2,2)$ entry (i.e.\ contract) is necessarily  $\pi_{22} = 1 - \pi_{11} - \pi_{12} - \pi_{21}$.}
    \label{fig:hyperprior_low_dim}
\end{wrapfigure}
Its parameters are $\mu_{0,1}\in(0,1)$, \;$\nu_{0,1}\in(0,1)$, \;  
$(\pi_{\O{I},11}, \pi_{\O{I},12}, \pi_{\O{I},21}) \in \Delta_3$ and the Kantorovitch potentials, $\boldsymbol{\lambda}^{o} \succeq \boldsymbol{0}$. 
The purpose of the following simulations is to study the influence of the Kantorovitch potentials, $\boldsymbol{\lambda}^{o}$ (\ref{eq:optKant}), and the nominal marginals, $\mu_{0}$ and $\nu_{0}$, on the location and shape of the hyperprior. 
For ease of visualization (in $\Delta_{2}$), we focus primarily on the bivariate marginal distribution\footnote{All integrals in this section are computed using  Gaussian quadrature integration, yielding results with an average integration error of $\approx 1.46 \times 10^{-8}$.} \eqref{eq:block_marginal}, i.e.\ the hyperprior concentrated on the two contracts forming the first row of the uncertain transport plan (Figure~\ref{fig:hyperprior_low_dim}): 
\begin{equation}\label{marginal_hyperprior}
\mathsf{S}^{o}(\pi_{11},\pi_{12}|K) \propto \int_{0}^{1-\pi_{11}-\pi_{12}}\mathsf{S}^{o}(\pi_{11},\pi_{12},\pi_{21}|K) d\pi_{21}
\end{equation}

\paragraph{Shape parameters:}
The cost matrix \eqref{general_gibbs_kernel} and nominal marginals are respectively set to the following values: 
\[
\O{C} \equiv 
\begin{bmatrix}
0 & 1\\
1 & 0
\end{bmatrix},\;\;\;
(\mu_{0}, \nu_{0}) \equiv \left\{\begin{bmatrix} 
0.3 \\
0.7
\end{bmatrix}, 
\begin{bmatrix} 
0.8 \\
0.2
\end{bmatrix} \right\}.
\]
For now, we fix the smoothness parameter $\epsilon = 1$ and study its influence on the shape of the hyperprior in a separate section. 
We examine the influence of the Kantorovitch potentials, $\boldsymbol{\lambda}^{o}$, on the shape of the hyperprior, by varying their values as follows: $\boldsymbol{\lambda}^{o} \in \{0.05, 10, 100\}^{2}$. 
\\\\
As discussed earlier, these potentials---through their connection to the KLD radii, $(\eta, \zeta)$---quantify the uncertainty in the marginals and induce two asymptotic learning modes. The first is attained when  $\boldsymbol{\lambda}^{o} \rightarrow \boldsymbol{0}$, and coincides with the non-specification of the marginals, and the  absence of effective learning. The second is attained when $\boldsymbol{\lambda}^{o} \rightarrow \infty$, i.e.\ when there is perfect specification of the marginals. The visualizations in Figure~\ref{fig:marginals_lambdas}---which shows the contour plots of the marginal hyperprior, $\mathsf{S}^{o}(\pi_{11},\pi_{12}|K)$ for the chosen values of  $\boldsymbol{\lambda}^{o}$---illustrate this concentration behaviour, as we progress from the first to the second modality. By increasing the potentials, the contours gradually concentrate on a thin statistical manifold, namely $\mathbb{\Pi}(\mu_{0}, \nu_{0})$. In addition to the marginal hyperprior, we show the first row, ($\pi_{11}$, $\pi_{12}$), of the expected transport plan, $\hat{\pi}_{\mathsf{S}}$ \eqref{eq:expected_plan} (red dot). The latter is obtained by averaging samples drawn from the joint hyperprior: $\pi^{(k)} \sim \mathsf{S}^{o}(\pi_{11}, \pi_{12}, \pi_{21}|K)$. The blue dot, on the other hand, corresponds to the first row of the EOT plan, $\pi^{o}(x_i,y_j |K)$ \eqref{eq:OT_solution}, computed for the nominal marginals, $(\mu_0,\nu_0)$, using the Sinkhorn-Knopp algorithm \citep{cuturi2013sinkhorn} (and is, of course, invariant with $\boldsymbol{\lambda}^{o}$). The expected transport plan gradually converges towards the OT plan, as the support of the marginal hyperprior contracts towards $\mathbb{\Pi}(\mu_{0}, \nu_{0})$ when $\boldsymbol{\lambda}^{o} \rightarrow \infty$, which is consistent with the Laplace-Bernstein concentration theorem.
\begin{figure}[htp]
\centering
\includegraphics[width=.3\textwidth]{ 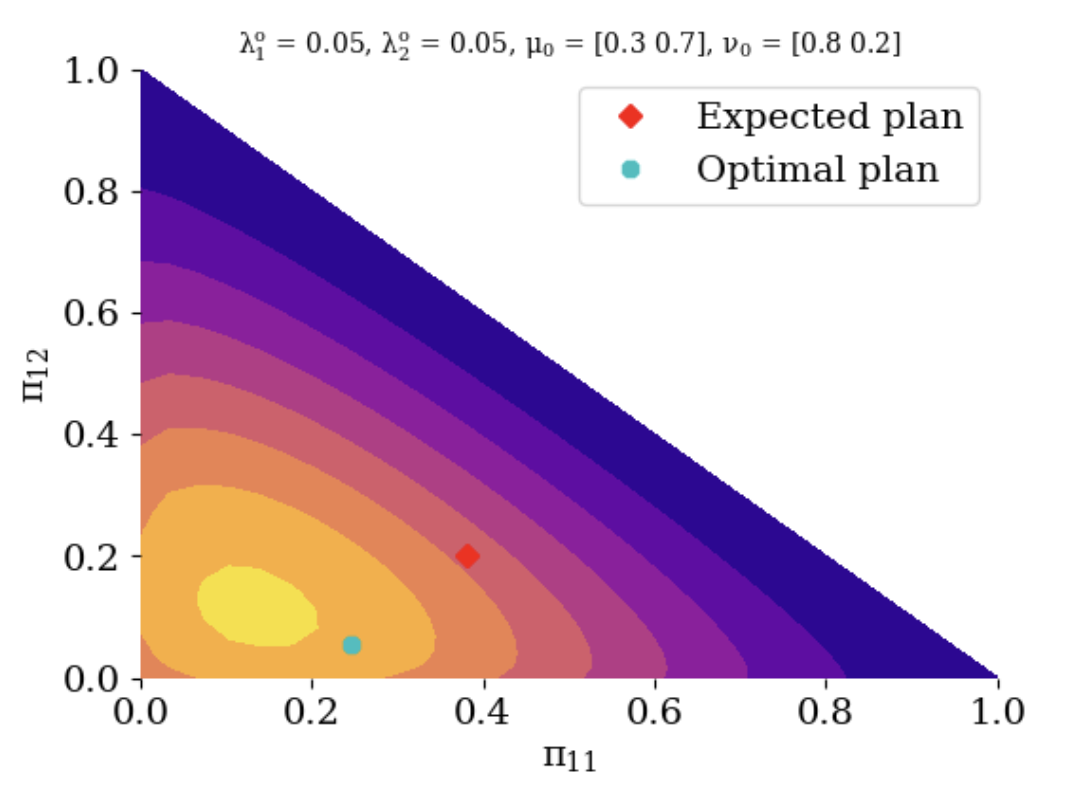}\hfill
\includegraphics[width=.3\textwidth]{ 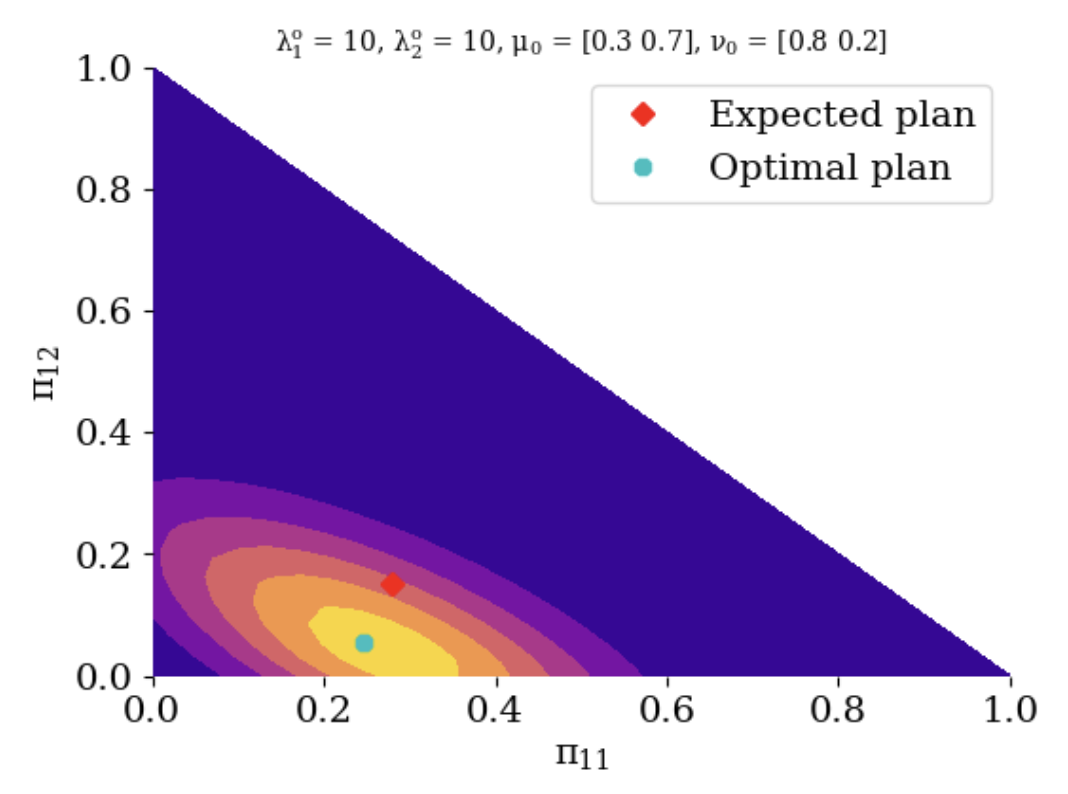}\hfill
\includegraphics[width=.3\textwidth]{ 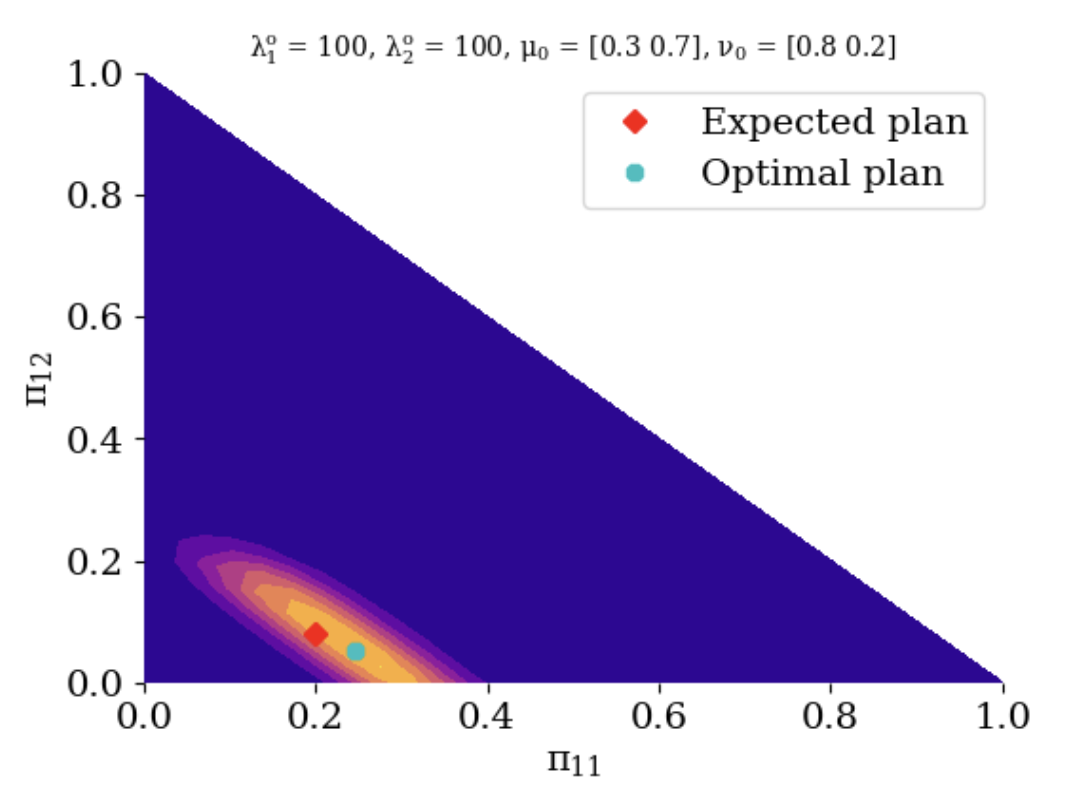}\hfill
\caption{ Contour plots of the bivariate marginal hyperprior, $\mathsf{S}^{o}(\pi_{11},\pi_{12}|K)$, defined over the 2D simplex, $\Delta_2$, for various values of the Kantorovitch potentials, 
$\boldsymbol{\lambda}^{o}$,
and for fixed nominal marginals, $(\mu_{0}, \nu_{0})$, and base-level ideal design. The \textcolor{red}{\bf red} dots correspond to the expected value of the first row, ($\pi_{11}$, $\pi_{12}$),  of the uncertain transport plan. We also show---via the \textcolor{blue}{\bf blue} dots---the first row of  the conventional EOT plan, for $(\mu_{0}, \nu_{0})$. }
\label{fig:marginals_lambdas}
\end{figure}
\paragraph{Location parameters:} The nominal marginals, $(\mu_{0}, \nu_{0})$, play the role of location parameters for the hyperprior. 
To illustrate this, we fix the Kantorovitch potentials and the smoothness parameter, respectively, to default values: $\boldsymbol{\lambda}^{o} = (1, 1)$, $\epsilon = 1$ and vary the nominal marginals as follows: 
\begin{equation}\label{eq:nominal_marginals_config}
(\mu_{0}, \nu_{0}) \in \left\{\begin{bmatrix} 
(0.9, 0.1) \\ (0.1, 0.9)
\end{bmatrix}^{\intercal}, 
\begin{bmatrix} 
(0.5, 0.5) \\ (0.5, 0.5)
\end{bmatrix}^{\intercal}, 
\begin{bmatrix} 
(0.1, 0.9) \\ (0.9, 0.1)
\end{bmatrix}^{\intercal}
 \right\}.
 \end{equation}
\\\\
For each pair of the nominal marginals in \eqref{eq:nominal_marginals_config}, we show in Figure \ref{marginals_nominals} the contour plot of the marginal hyperprior, $\mathsf{S}^{o}(\pi_{11},\pi_{12}|K)$. Moreover, we plot the first row of the expected transport plan, $\hat{\pi}_{\mathsf{S}}$, in red and the EOT plan, $\pi^{o}$ in blue. The location of the mode is clearly influenced by the nominal marginals, $(\mu_{0}, \nu_{0})$, and more precisely, by their symmetry and skewness. The expected plan, $\hat{\pi}_{\mathsf{S}}$ \eqref{eq:expected_plan}, is  attracted by the mode of the marginal hyperprior;  the optimal plan, on the other hand, initially has a low probability under the marginal hyperprior but contracts gradually towards the mode. 

\begin{figure}[htp]
\centering
\includegraphics[width=.3\textwidth]{ 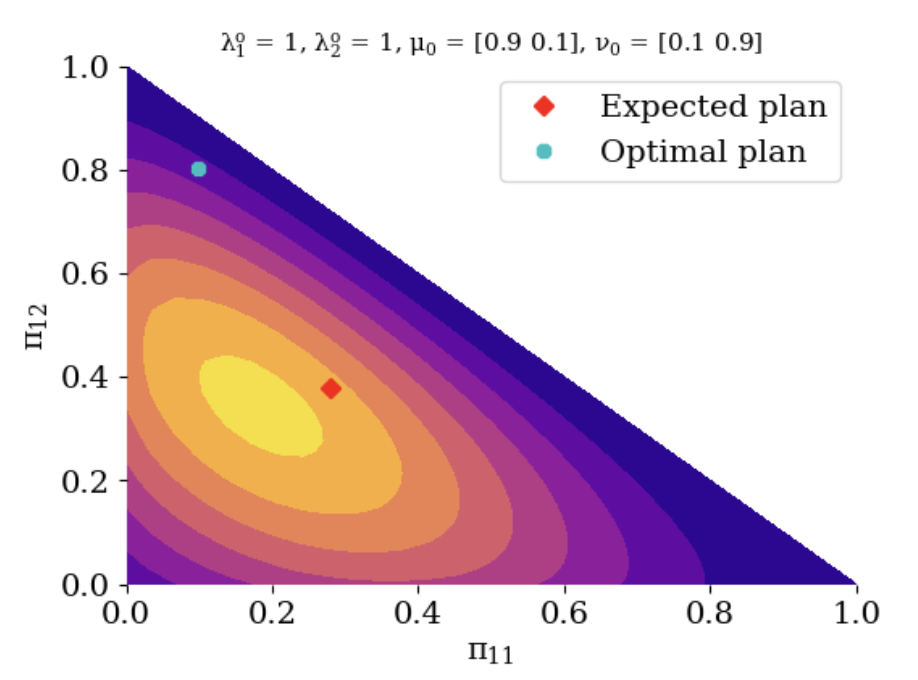}\hfill
\includegraphics[width=.3\textwidth]{ 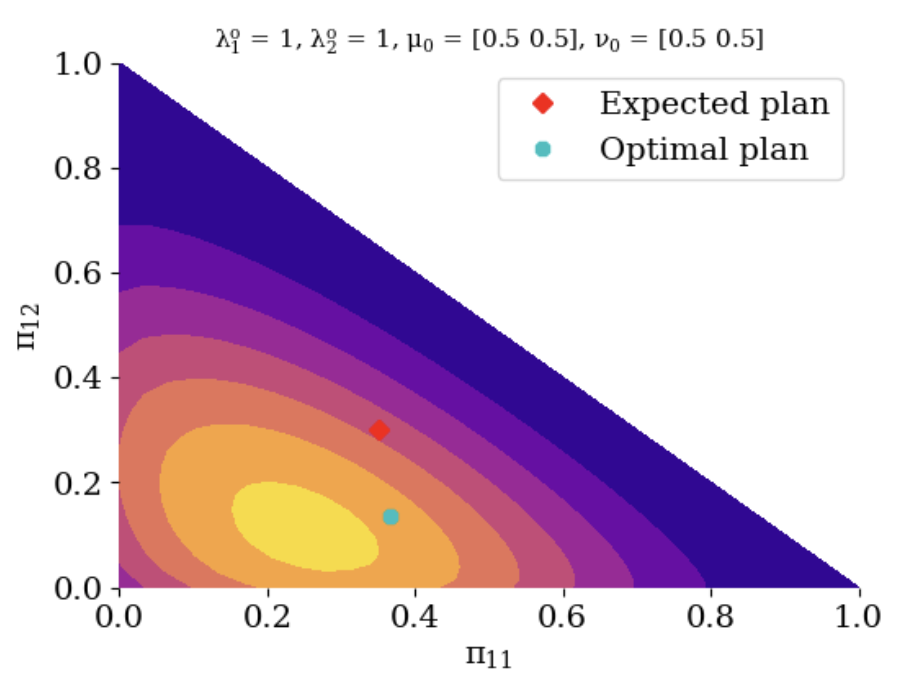}\hfill
\includegraphics[width=.3\textwidth]{ 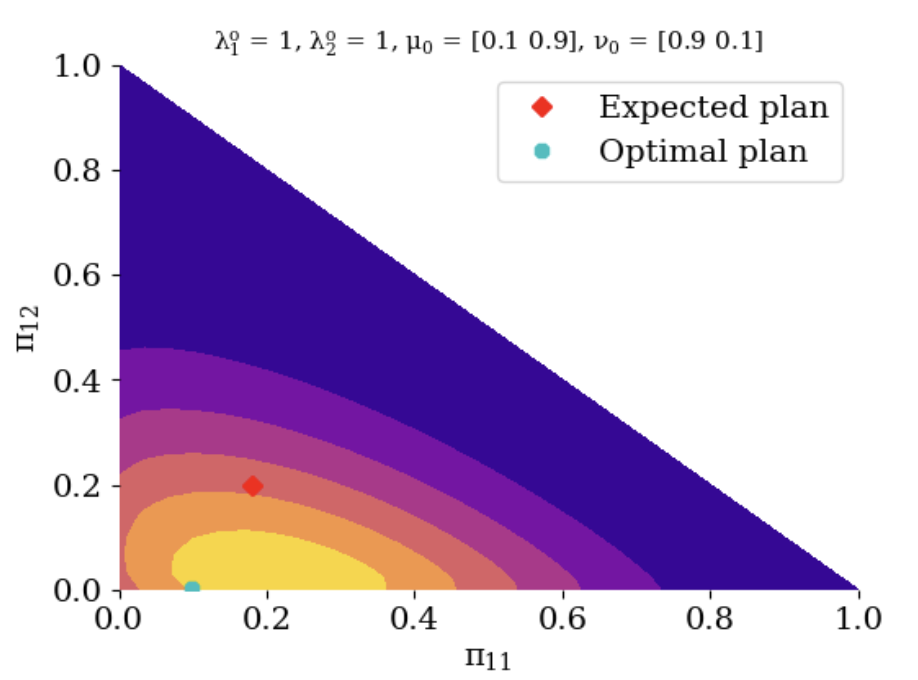}\hfill
\caption{ Marginal hyperprior, $\mathsf{S}^{o}(\pi_{11}, \pi_{12}|K)$, for fixed Kantorovitch potentials, $\boldsymbol{\lambda}^o$, and various values of the nominal marginals, $(\mu_{0}, \nu_{0})$. The \textcolor{red}{\textbf{red}} and \textcolor{blue}{\textbf{blue}} dots  correspond to the first row  of the expected transport plan, and of the EOT plan,  respectively.}
\label{marginals_nominals}
\end{figure}
\paragraph{Influence of the ideal hyperprior \label{par:influence_ideal}:} 
Finally, we explore the influence of the ideal design  \eqref{eq:joint_ideal_design}, and, more precisely, its smoothness parameter, $\epsilon$, which enters at  the base-level of the ideal specification \eqref{general_gibbs_kernel}.  We hold the nominal marginals, $(\mu_0,\nu_0)$, constant, as indicated.  By varying $\epsilon \in \{0.1, 0.5, 10\}$, it is clear from 
Figure \ref{marginals_epsilon} that this parameter  affects the location of the hyperprior,  $\mathsf{S}^{o}(\pi_{11}, \pi_{12}|K)$. 
\begin{figure}[htp]
\centering
\includegraphics[width=.3\textwidth]{ 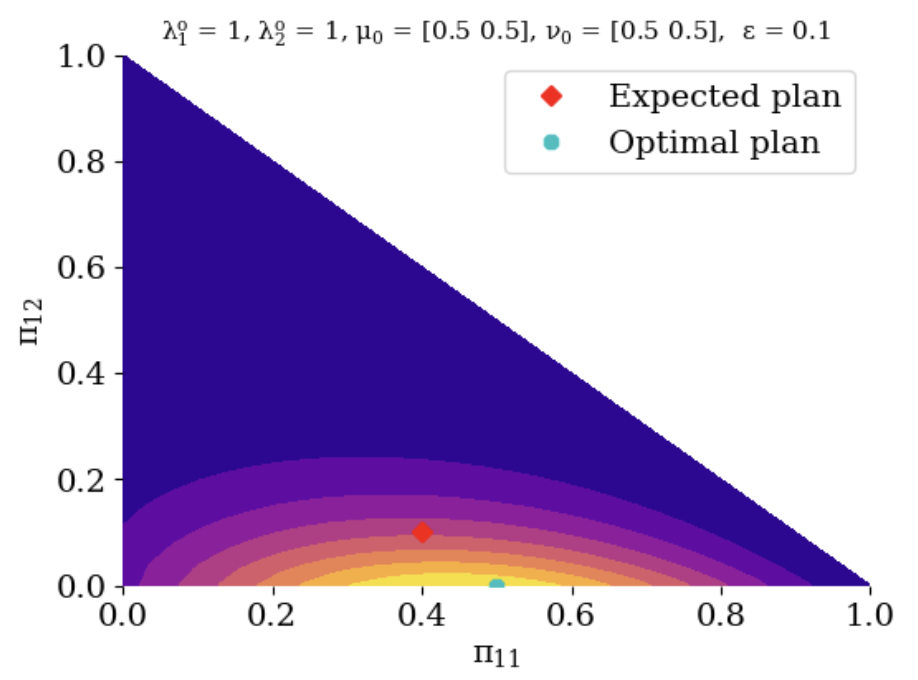}\hfill
\includegraphics[width=.3\textwidth]{ 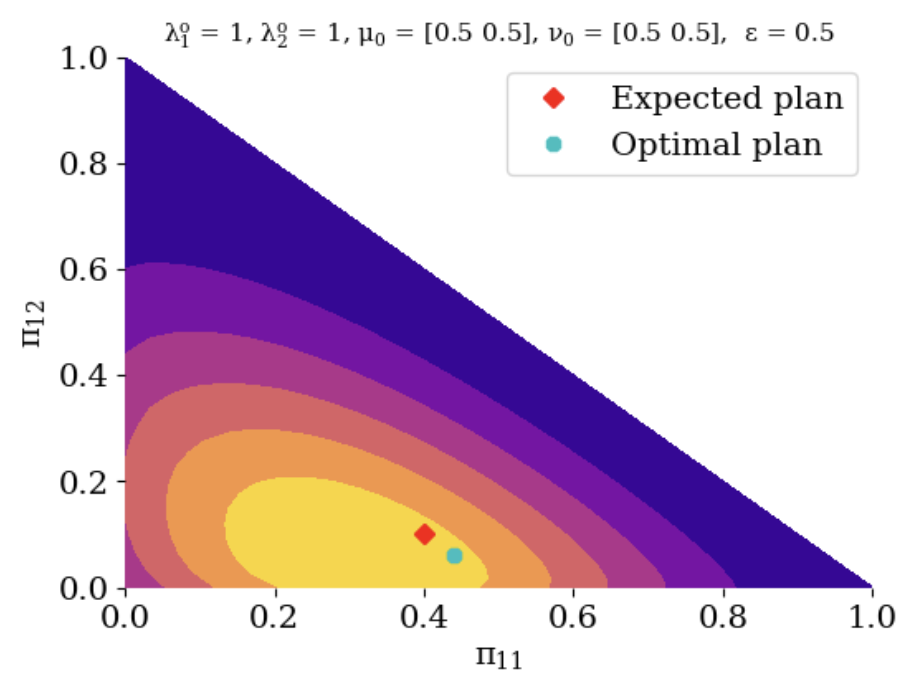}\hfill
\includegraphics[width=.3\textwidth]{ 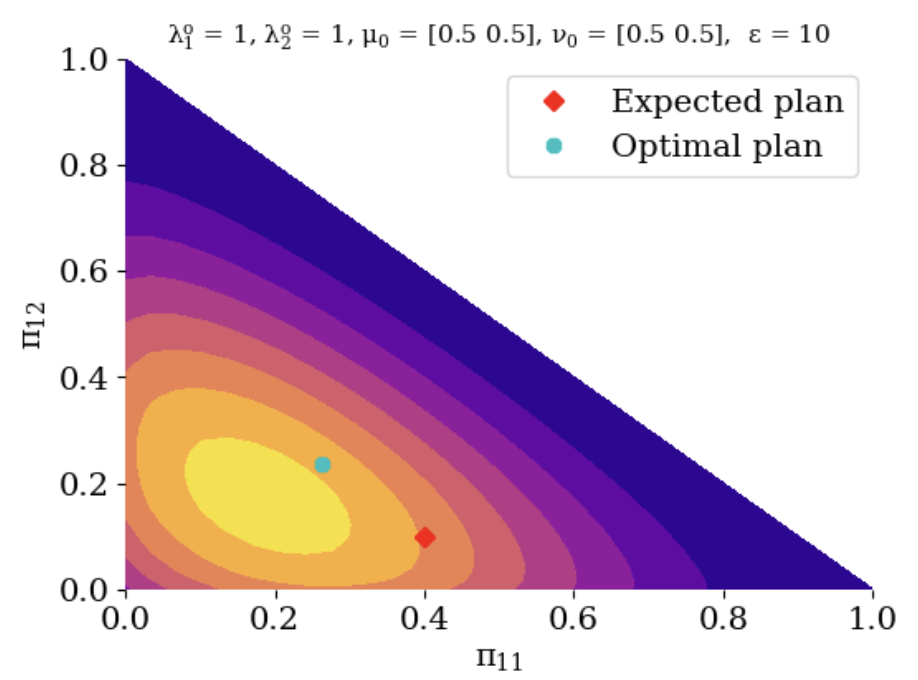}\hfill
\caption{Contour plots of the marginal hyperprior $\mathsf{S}^{o}(\pi_{11}, \pi_{12}|K)$, for various values of the smoothness parameter $\epsilon$. The \textcolor{red}{\textbf{red}} and \textcolor{blue}{\textbf{blue}} dots  correspond to the first row  of the expected transport plan, and of the EOT plan,  respectively.}
\label{marginals_epsilon}
\end{figure}
%

\subsection{Stochastic approximation of the optimal Kantorovitch potentials}\label{Design_Kantorovitch_Potentials}

We now focus on the derivation of the optimal Kantorovitch potentials, $\boldsymbol{\lambda}^{o}$. This requires processing the knowledge constraints, $(\eta, \zeta)$, in the hyperprior, by solving the dual program \eqref{eq:dual}. To this end, we leverage a combination of second-order optimization and MCMC techniques.

Computing $\boldsymbol{\lambda}^{o}$ by means of the dual program in \eqref{eq:dual} is a critical step in the design of the optimal hyperprior, $\mathsf{S}^{o}(\pi|K)$ \eqref{eq:parametric_hyperprior}. However, deriving their exact values in high-dimensional settings is not trivial, as it requires manipulating the intractable normalizing constant \eqref{normalization}. The methodology proposed herein approximates these potentials using a combination of Quasi-Newton \citep{NoceWrig06}, \citep{nesterov} and Hamiltonian Monte Carlo (HMC) \citep{betancourt2018conceptual}, thus circumventing the need to evaluate $\mathsf{N}(\boldsymbol{\lambda})$. In particular, HMC provides a rigorous and efficient framework for sampling in high-dimensional settings: compared to other MCMC techniques, the number of gradient estimations in HMC is less sensitive to the dimension of the problem~\citep{pmlr-v89-mangoubi19a}, making it a convenient choice when generating random transport plans $\pi^{(k)}\sim\O{S}^{o}$. \\\\ 
As proved in \eqref{eq:dual}, the optimal Kantorovitch potentials read as follows:
\begin{equation}\label{eq:optimal_potentials}
\boldsymbol{\lambda}^{o} = \argmin_{\boldsymbol{\lambda} \succeq 0} \left\{ \boldsymbol{\lambda}^{\intercal}\boldsymbol{\theta} - \log\left(\mathsf{N}(\boldsymbol{\lambda})\right) \right\}.
\end{equation}
Let $\varrho(\boldsymbol{\lambda}) \equiv \boldsymbol{\lambda}^{\intercal}\boldsymbol{\theta} - \log\left(\mathsf{N}(\boldsymbol{\lambda})\right) $ denote the optimization objective in \eqref{eq:optimal_potentials}. Its gradient vector can be written conveniently using the following expectation:
\begin{equation}\label{gradient}
\begin{split}
\nabla_{\boldsymbol{\lambda}} \varrho(\boldsymbol{\lambda}) &= \boldsymbol{\theta} - \O{E}_{\mathsf{S}}\bigl[\mathsf{R}(\pi)\bigr]
\end{split}
\end{equation}
We define $s_{t}$ and $n_{t}$, the Kantorovitch potentials and their gradient differentials respectively, as follows: 
\[s_{t} \equiv\boldsymbol{\lambda}_{t+1}-\boldsymbol{\lambda}_{t}\] and \[n_{t} \equiv\nabla\varrho(\boldsymbol{\lambda}_{t+1})-\nabla\varrho(\boldsymbol{\lambda}_{t})\]
where $t>0$ is the iteration in Quasi-Newton. The recursive approximation of the inverse Hessian can be written as follows \citep{NoceWrig06}:
\begin{equation}
    \mathsf{H}_{t+1} = (\boldsymbol{\mathsf{I}}-\varsigma_{t}s_{t}n_{t}^{\intercal})\mathsf{H}_{t}(\boldsymbol{\mathsf{I}}-\varsigma_{t}n_{t}s_{t}^{\intercal}) + \varsigma_{t}s_{t}s_{t}^{\intercal}  \;\;\;, \;\ \varsigma_{t} \equiv\frac{1}{n_{t}^{\intercal}s_{t}}
\end{equation}
where $\boldsymbol{\mathsf{I}}$ denotes the identity matrix. We note that the inverse Hessian $\mathsf{H}_{t}$ depends only on the stochastic gradients $\nabla\varrho(\boldsymbol{\lambda}_{t})$ \eqref{gradient}. Thus, we avoid stability issues when dealing with ill-conditioned stochastic inverse Hessian approximations, as it is the case with high-variance MC samplers. \\\\
Once computed, the gradient and the inverse Hessian are plugged into the usual BFGS iterative updates \citep{NoceWrig06}:
\begin{equation}
\boldsymbol{\lambda}_{t+1} \xleftarrow{} \boldsymbol{\lambda}_{t} - \rho_{t} \mathsf{H}_{t} \nabla_{\boldsymbol{\lambda}} \varrho(\boldsymbol{\lambda}_{t})
\end{equation}
where $\rho_{t} > 0$ is the step size at the $t^{th}$ iteration in the search direction given by: 
\begin{equation} \label{search_dir}
\mathsf{d}(\boldsymbol{\lambda}_{t}) \equiv- \mathsf{H}_{t} \nabla_{\boldsymbol{\lambda}} \varrho(\boldsymbol{\lambda}_{t})
\end{equation}
The step size $\rho_{t}$ should be adapted carefully to ensure convergence to the global minimum $\boldsymbol{\lambda}^{o}$. It is usually computed by solving an auxiliary line search problem, using techniques such as backtrack line search (BTLS) \citep{nesterov}. However, most of line search techniques require the evaluation of the objective $\varrho(\boldsymbol{\lambda})$ at each step. To avoid explicit function evaluations, we propose a simple local approximation that estimates the position of the minimum along the search line \eqref{search_dir}, based solely on two gradient evaluations \citep{SnymanGradientLineSearch}. \\\\
More precisely, the optimal step size that yields sufficient decrease in the search direction \eqref{search_dir} can be found by solving the following problem:
\begin{equation}\label{eq:line_search}
\rho_{t}^{*} = \argmin_{\rho \in [0, 1] }\varrho(\boldsymbol{\lambda}_{t} + \rho \; \mathsf{d}(\boldsymbol{\lambda}_{t}))
\end{equation}
Assuming that $\varrho$ is locally quadratic at $\boldsymbol{\lambda}_{t}$, it follows that solving \eqref{eq:line_search} reduces to finding $\rho$ that satisfies:
\begin{equation}
\varrho(\boldsymbol{\lambda}_{t} + \rho \; \mathsf{d}(\boldsymbol{\lambda}_{t})) = \varrho(\boldsymbol{\lambda}_{t})
\end{equation}
Which yields the following optimal step size:
\begin{equation}\label{eq:optimal_step_size}
\begin{split}
    \rho_{t}^{*} &= \frac{- \mathsf{d}(\boldsymbol{\lambda}_{t})^{\intercal}\nabla \varrho(\boldsymbol{\lambda}_{t})}{{\mathsf{d}(\boldsymbol{\lambda}_{t})}^{\intercal} \nabla^{2} \varrho(\boldsymbol{\lambda}_{t}) \mathsf{d}(\boldsymbol{\lambda}_{t})}
\end{split}
\end{equation}
Finally, by a second-order Taylor expansion at $\boldsymbol{\lambda}_{t}$ and $\boldsymbol{\lambda}_{t} + \mathsf{d}(\boldsymbol{\lambda}_{t})$, the denominator in \eqref{eq:optimal_step_size} can be computed using two gradients estimations, as follows:
\begin{equation}
{\mathsf{d}(\boldsymbol{\lambda}_{t})}^{\intercal}\nabla^{2}\varrho(\boldsymbol{\lambda}_{t})\mathsf{d}(\boldsymbol{\lambda}_{t}) \approx {\mathsf{d}(\boldsymbol{\lambda}_{t})}^{\intercal}\left[\nabla\varrho(\boldsymbol{\lambda}_{t} + \mathsf{d}(\boldsymbol{\lambda}_{t})) -  \nabla\varrho(\boldsymbol{\lambda}_{t})\right] 
\end{equation}

What remains is to compute the gradient terms, which can be estimated using HMC. If $n_{s} > 0$ is the number of independent realizations $\pi^{(i)} \sim \O{S}(\pi|K)$, then the expectation in \eqref{gradient} can be approximated as follows:
\begin{equation}\label{expectation}
\O{E}_{\mathsf{S}}\bigl[\mathsf{R}(\pi)\bigr] \approx \frac{1}{n_{samp}} \sum \limits_{i=1}^{n_{s}} \O{R}(\pi^{(i)})
\end{equation}
At each iteration $t$, the error (stopping criterion) is measured by means of the following Newton's decrement, which corresponds to the inverse Hessian norm of the gradient. This quantity provides a good indication of the proximity to the optimal Kantorovitch potentials:
\begin{equation}
\mathsf{err} \equiv\nabla_{\boldsymbol{\lambda}}\varrho(\boldsymbol{\lambda}_{t})^{\intercal}\nabla_{\boldsymbol{\lambda}}^{2}\varrho(\boldsymbol{\lambda}_{t})^{-1}\nabla_{\boldsymbol{\lambda}}\varrho(\boldsymbol{\lambda}_{t})
\end{equation}
The optimal potentials $\boldsymbol{\lambda}^{o}$ are then plugged into \eqref{eq:hyperprior} and the optimal hyperprior can be used to generate random transport plans, by means of another HMC sampler.

\RestyleAlgo{ruled}
\SetAlgoVlined
\LinesNumbered
\SetKwInput{KwInput}{Input}
\begin{algorithm}[H]
\caption{Approximation of the Kantorovitch potentials} \label{algo:compute_kanto_algo}
  \KwInput{nominal marginals $(\mu_{0}, \nu_{0})$, KLD radii $(\eta, \zeta)$, target precision $\tau$, base-level ideal design $\pi_{\mathsf{I}}$, hierarchical ideal design $\mathsf{S}_{\mathsf{I}}$, number of samples $n_\O{samp}$}
  \KwResult{$\boldsymbol{\lambda}^{o}$}
  Initialization: $t=1$, $\boldsymbol{\lambda}_{t} \succeq 0$, $\rho_{t}=1$, $\mathsf{H}_{t}=\boldsymbol{\mathsf{I}}$, $\mathsf{err}=\infty$ \;
  \While{$\tau < \mathsf{err}$}{
    \text{Sample} $\{\pi^{(l)}_{t}\}_{l=1}^{n_\O{samp}} \sim \mathsf{S}(\pi|K_{-\boldsymbol{\lambda}}, \boldsymbol{\lambda}_{t})$ \algorithmiccomment{\textit{HMC sampler. $K_{-\boldsymbol{\lambda}}$ denotes all parameters in the knowledge set $K$, except $\boldsymbol{\lambda}$}} \;
    \text{Estimate} $\O{E}_{\mathsf{S}(\pi|K_{-\boldsymbol{\lambda}}, \boldsymbol{\lambda}_{t})}\left[\mathsf{R}(\pi)\right]$ \;
    \text{Estimate} $\nabla_{\boldsymbol{\lambda}}\varrho(\boldsymbol{\lambda}_{t})$ \;
    \text{Compute} $\check{\boldsymbol{\lambda}}_{t+1} \leftarrow \check{\boldsymbol{\lambda}}_{t} - \mathsf{H}_{t} \nabla \varrho(\boldsymbol{\lambda}_{t})$ \;
    \text{Sample} $\{\check{\pi}^{(l)}_{t+1}\}_{l=1}^{n_\O{samp}} \sim \mathsf{S}(\pi|K_{-\boldsymbol{\lambda}}, \check{\boldsymbol{\lambda}}_{t+1})$ \;
    \text{Estimate} $\nabla_{\boldsymbol{\lambda}} \varrho(\check{\boldsymbol{\lambda}}_{t+1})$ \;
    \text{Compute} $\rho^{*}_{t}$ \;
    \text{Update}  $\boldsymbol{\lambda}_{t+1} \xleftarrow{} \boldsymbol{\lambda}_{t} - \rho^{*}_{t} \mathsf{H}_{t} \nabla \varrho(\boldsymbol{\lambda}_{t})$ \;
    \text{Compute}  $s_{t}$, $n_{t}$ and $\varsigma_{t}$ \;
    \text{Update} $ \mathsf{H}_{t+1} \leftarrow (\boldsymbol{\mathsf{I}}-\varsigma_{t}s_{t}n_{t}^{\intercal})\mathsf{H}_{t}(\boldsymbol{\mathsf{I}}-\varsigma_{t}n_{t}s_{t}^{\intercal}) + \varsigma_{t}s_{t}s_{t}^{\intercal}$ \;
    \text{Update} $\mathsf{err}$ \;
    \text{Update} $t \leftarrow t+1$
  }
 \Return $\boldsymbol{\lambda}_{t+1}$
\end{algorithm}
\begin{remark}\label{rmk:computational_complexity}\textbf{Computational complexity.}
In  Algorithm~\ref{algo:compute_kanto_algo}, we replace each approximation of the normalising constant, $\mathsf{N}(\boldsymbol{\lambda})$ \eqref{normalization}, with two gradient approximations. Therefore, the overall computational complexity is mainly driven by the sampling operations in line $3$ and $7$ of the Algorithm, whose complexity is, in turn, contingent upon the number of gradient evaluations used in the leapfrog integrator of the HMC sampler \citep{betancourt2018conceptual}. Under certain regularity conditions, this number is of order $\mathcal{O}(\sqrt{mn})$ \citep{pmlr-v89-mangoubi19a}. Though these regularity conditions are not fully satisfied here (see Remark \ref{rmk:hmc_mixing_time}), this provides us with a good lower bound on the computational complexity. Using a mean-field variational Bayes method at each iteration of the Quasi-Newton method---which assumes that all the parameters (i.e. contracts), $\pi_{i,j}$, are independent---would result in a linear time complexity in the number of parameters, that is $\mathcal{O}(mn)$.
\end{remark}
\begin{remark}\label{rmk:hmc_mixing_time} \textbf{On HMC mixing properties.}
It is worth noting that the main convergence results of HMC, when sampling from a log-concave function, $\O{e}^{-f}$, require strongly convex and Lipschitz smooth (i.e. Lipschitz $\nabla f$) potential functions, $f$ \citep{chen2019optimalconvergenceratehamiltonian}. However, the KLD is not Lipschitz smooth and the theoretical convergence results are not guaranteed in our setting. This results in a longer integration time and biased estimators, especially when $(\eta, \zeta) \rightarrow (0, 0)$. For the time being, we will use HMC while carefully tuning its main parameters (integrator step size, adaptation step, {\em etc}.), and will explore specialized samplers in a separate work.
\end{remark}
\section{HFPD-OT for Algorithmic fairness in market matching}\label{simulations}

The goal of algorithmic fairness is to detect and mitigate algorithmic biases induced by automated decision-making systems \citep{Barocas2018FairnessAM}. This is a compelling setting for HFPD-OT, since we can benefit from randomized transport plans to elicit fair transport strategies in the presence of uncertainty. Note that OT {\em for\/} fairness has already been proposed in other works (see \citep{delbarrio2018obtaining} and references therein), with the focus being on notions of data repair and learning fair models. In contrast, we are concerned, here,   with {\em fair OT}, whose purpose is to design transport plans that are fair {\em per se}. The literature on fair OT is  sparse: in \citep{huguesjasonchen2021}, the authors address the fair OT problem by proposing a dynamic and distributed fair OT algorithm. In this manuscript, we propose a different approach that leverages randomized policies, which are induced naturally by the HFPD-OT setting.

To appreciate the implications for fair OT of the randomization and diversity allowed by HFPD-OT, we study the problem of fair market matching \citep{galichon2021unreasonableeffectivenessoptimaltransport}, \citep{echenique2024stablematchingtransportation}, and more precisely the question of worker-job matching, in which the nominal marginals, $\mu_0$ and $\nu_0$, are estimates of the distributions of workers and jobs, respectively. An agent $x_{i} \in \mathbb{\Omega}_{X}$ represents a category of workers or skills, while an agent $y_{j} \in \mathbb{\Omega}_{Y}$ is a job opportunity or a company. In particular, we study  \textit{vertically}-differentiated agents: workers in one category may exhibit skills not available in other categories. Similarly, some companies may differ in their size or may have unique production technologies \footnote{This is in contrast to  \textit{horizontally} differentiated agents, where some hierarchy may exist between agents.}. A contract $\pi_{i,j}$ seeks to match some of  the workers in category $x_{i}$ with some of the job opportunities offered by $y_{j}$. 

Our purpose is to study the following question: 
    \textit{Can randomized transport plans elicit long-term fair matching strategies in a worker-job matching problem, for agents as well as for individual contracts?}
Our notion of fairness is asymptotic, in the sense that fairness is achieved in the long-run. This is in contrast to the static (i.e.\ invariant) designs of  classical OT,  which may, indeed, satisfy a standard fairness metric based on the ensemble of contracts on $\mathbb{\Omega}_{X} \times \mathbb{\Omega}_{Y}$, but, unfortunately, harms the same individual agents or contracts, either because of:
\begin{itemize}
\item[(i)] misspecification of the marginals for some of the agents, $\mathbb{\Omega}_{X}$ or $\mathbb{\Omega}_{Y}$; and/or 
\item[(ii)] an invariant and uneven distribution of mass across the contracts, $\pi_{i,j}^o$. 
\end{itemize}
Before addressing the problem of fair labour market matching (Section~\ref{sec:marketmatching}), we review the fairness-related concept of diversity.  

\subsection{Simulation study}
\label{sec:sim}
We consider the following setting: 
\begin{itemize}\label{experimental_setting}
\item $m \equiv n \equiv d \equiv 20$
\item $\epsilon \equiv 10^{-3}$
\item $\mathsf{C}_{i, j} \equiv \|x_{i}-y_{j}\|_2^{2}, \;\; (i , j) \in [\![m]\!] \times [\![n]\!] \ $
\item $\eta \equiv 2$,\; $\zeta\equiv 2$
\item $\lambda_{\O{I}, 1} \equiv 0.5$,\; $\lambda_{\O{I}, 2} \equiv 0.5$
We simulate the nominal worker and job distributions as  $\mu_{0} \sim \mathcal{tN}(2, 5)$ and $\nu_{0} \sim \mathcal{tN}(6, 3)$, respectively, where $\mathcal{tN}(a, b)$ denotes the truncated Gaussian distribution with positive support,  mean $a$ and variance $b$.
\item To sample from the hyperprior, $\O{S}^o( \pi | K)$ (\ref{eq:parametric_hyperprior}),  we leverage the Hamiltonian Monte Carlo (HMC) module available in TensorFlow Probability (version 0.24.0)\footnote{https://www.tensorflow.org/probability}, with the following configuration: 
\begin{itemize}
    \item Number of burn-in steps: $8000$
    \item Number of adaptation steps: $0.8 \times \text{number of burn-in steps}$
    \item Target acceptance probability (fixed):  $0.6$
    \item The length traveled by the leapfrog integrator is adjusted using a No U-Turn Sampler (NUTS) \citep{hoffman2011nouturnsampleradaptivelysetting}.
    \item The step size is optimized using a dual averaging policy \citep{hoffman2011nouturnsampleradaptivelysetting}.
    \item The sampler is compiled using XLA (Accelerated Linear Algebra).
    \item The optimal Kantorovitch potentials \eqref{eq:optimal_potentials} are computed using Algorithm~\eqref{algo:compute_kanto_algo}.
\end{itemize} 
\item The base-level EOT model (\ref{eq:general_fpd_ot}) is computed using the POT library \citep{10.5555/3546258.3546336}.
\end{itemize}

\subsection{Quantifying diversity in HFPD-OT}
Our definition of long-term fairness---to follow---relies on the notion of diversity, which we quantify using the following \textbf{diversity index}:

\vspace*{.2cm}

\begin{definition}[Diversity index]
\label{def:div}
    Let  $m \times n$ be the dimension of the parametric random transport plan, $\pi\sim \O{S}^{o}(\pi|K)$ (\ref{eq:parametric_hyperprior}).  The 1-diversity index (or perplexity score \citep{Jelinek1977PerplexityaMO})  associated with $\O{S}^{o}$ is:
    \begin{equation}\label{eq:diversity_index}
\O{D}(\O{S}^{o}(\pi|K)) \equiv\O{E}_{\O{S}^{o}}\left[\exp\bigl(\O{H}(\pi)\bigr)\right],
    \end{equation}
where $\O{H}(\pi)$ denotes the entropy of $\pi$:
\begin{equation}
\O{H}(\pi) \equiv-\sum_{i=1}^{m} \sum_{j=1}^{n} \pi_{i,j}\log(\pi_{i,j})
\end{equation}
\end{definition}
    
\begin{figure}[htp]
\begin{subfigure}{1.0\linewidth}
\centering
\includegraphics[width=.5\textwidth]{ 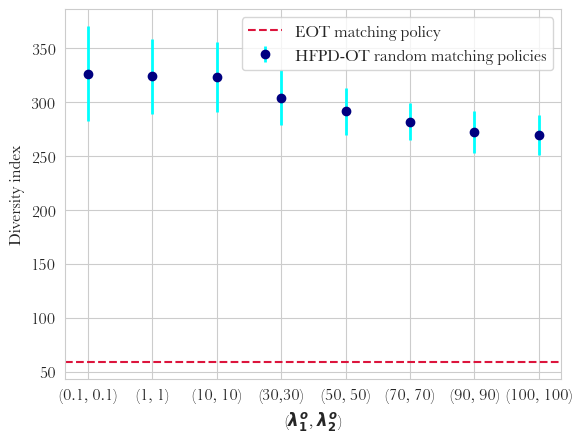}\hspace{0.7cm}
\caption{Diversity index, $\O{D}(\cdot)$, computed for different values of the Kantorovitch potentials. The average diversity index attained by HFPD-OT (red dots) remains greater than that of the EOT policy (red line), even when the latter is computed using a relatively high smoothing parameter, $\epsilon = 0.1$.}
\label{fig:fig781}
\end{subfigure}
\begin{subfigure}{1.0\linewidth}
\centering
\includegraphics[width=.5\textwidth]{ 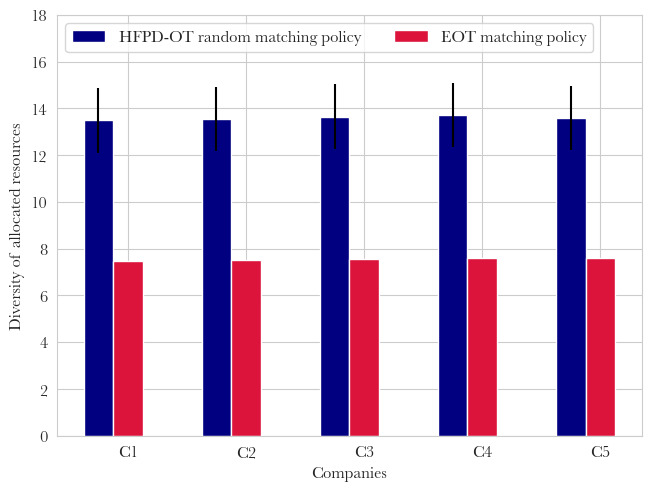}\hfill
\caption{Fairness for agents illustrated by computing the mean diversity index of the conditional transport plan $\pi(.|Y=y_{0})$ for five different companies $(C_{1}, \dots, C_{5})$. The Kantorovitch potentials and the smoothing parameter are respectively fixed to: $(\lambda_{1}^{o}, \lambda_{2}^{o})= (10, 10)$, and $\epsilon = 0.1$. Here again, a high diversity index means that each company is matched to a far more diverse set of skills and workers than it would be possible with highly-smoothed EOT policies.}
\label{fig:fig782}
\end{subfigure}
\caption{Comparative study of the diversity, $\O{D}(\cdot)$, induced by HFPD-OT random matching policies and fixed EOT matching policies. Error bars correspond to the $95\%$ confidence interval over 100 Monte Carlo experiments in HFPD-OT.}
\end{figure}
In Figure \ref{fig:fig781}, we compute and graph $\O{D}(\cdot)$ for different values of the Kantorovitch potentials (\ref{eq:optimal_potentials}), and we compare the diversity of random HFPD-OT matching polices to that of the base-level EOT policy (\ref{eq:general_fpd_ot}). While increasing the Kantorovitch potentials decreases the diversity, it remains substantially higher than that of the EOT policy, even when the smoothness parameter is fixed at a relatively high value: $\epsilon = 0.1$. In practical terms, a higher $\O{D}(\cdot)$ ensures that a more diverse set of skills is allocated to each company, in expectation. Similarly, workers are expected to have access to a more diverse set of job opportunities.  We use this insight in the sequel, to formalize the meanings of diversity and fairness both for agents (Definition~\ref{def:fairness_def1}) and contracts (Definition~\ref{def:fairness_def2}).
\begin{remark}
One might argue that the smoothness parameter, $\epsilon >0$,  in base-level EOT (\ref{eq:general_fpd_ot}) can be used to induce some level of diversity for fair OT (i.e.\ objective (ii) in Section~\ref{simulations}). However, it does not address objective (i). Note that the randomness in HFPD-OT is {\em informed}, since it emerges from modelling the uncertainty in the marginals, whereas the smoothness in EOT is mainly a computational convenience that is not informed by a mathematical model of uncertainty. 
\end{remark}
\subsection{Long-term fairness for agents through randomization}
We first discuss the notion of fairness for agents (groups of workers and companies, in our application) enabled by a randomized transport strategy and propose the following definition. 

\vspace*{.2cm}

\begin{definition}\textbf{Fairness for agents via diversified transport plans}\label{def:fairness_def1}
    
    A transport policy fulfills fairness for agents if: 
\begin{enumerate}
\item  It acknowledges that marginals (the supply and demand) may not have been fairly measured and incorporates this knowledge in the design of the transport policy.
\item It allows asymptomatically for a diversified allocation of resources. This diversification should be proportional to the uncertainty in the marginals.
\end{enumerate}
\end{definition}


Underestimating the supply of a category of workers can produce a matching policy in which all workers in that category are unfairly assigned to closer companies (in the sense of the cost $\O{C}$). Accounting for uncertainty in the supply, however, would allow,  in expectation, for a more diverse mix of skills to be transferred to companies. To illustrate this point further, we analyze the diversity of workers 
matched to companies and compute the mean diversity index of the {\em random\/} conditional transport policy $\pi(.|Y=y_{0})$ associated with each company, $y_{0} \in \mathbb{\Omega}_{Y}$. Figure \ref{fig:fig782} shows that the diversity of skills allowed by HFPD-OT remains consistently higher than that of the base-level EOT, thus allowing each company $y_{0}$ to benefit from a more diverse set of skills.

\subsection{Long-term fairness for contracts through randomization}
\label{sec:marketmatching}
In our worker-company matching problem, as in many other transport problems, contracts correspond to a physical infrastructure, deployed to match resources to demand (agencies, recruitment processes, crowd-sourcing labour market platforms, {\em etc}.). By design, the OT model yields a sparse transport strategy where the transport burden is supported by a small number of contracts, and though the base-level EOT may allow for smoother, i.e. more diverse transport strategies,  this diversity does not emerge from a proper mathematical modelling of uncertainty (Remark \ref{remark:randomness}). In contrast, randomized HFPD-OT strategies allow the activation of a more diverse set of contracts, yielding a fairer utilization of the transport infrastructure. In this regard, HFPD-OT is closely related to maximum diversity problems \citep{MARTI2022795}.  \\\\
To formalize the previous point, we start by introducing the notion of {\em eligible} contracts:
\begin{equation}\label{eq:eligible_contracts}
\mathbb{\Pi}_{\O{E}}(\eta, \zeta, \upsilon) \equiv \bigl\{\pi_{i,j}, \; (i , j) \in [\![m]\!] \times [\![n]\!] \; | \; \pi \in \O{supp}(\O{S}^{o}) \; \text{and} \; \O{E}_{\O{S}^{o}}\left[\mathbb{1}(\pi_{i,j} \geq \upsilon)\right] > 0\bigr\}.
\end{equation}
Here, $\upsilon > 0$ is an activation threshold, imposed by design constraints (technical specifications, design requirements, {\em etc}.). Eligible contracts are those with a positive probability of being active under the hyperprior, $\O{S}^{o}(\pi|K)$.  The set $\mathbb{\Pi}_{\O{E}}$ is better understood through its asymptotic behaviour:
\begin{itemize}
\item In the absence of any constraint on the marginals, $\mathbb{\Pi}_{\O{E}}$ is fully determined by the base-level and hierarchical ideal designs (\ref{eq:unspecified_marginals}), and:
\begin{equation} \mathbb{\Pi}_{\O{E}}(\eta, \zeta, \upsilon) \xrightarrow{\eta \rightarrow \infty, \zeta \rightarrow \infty} \bigl\{\pi_{i,j}, \; (i,j) \in [\![m]\!]\times [\![n]\!] \; | \;  \pi \in \O{supp}(\tilde{\O{S}})  \; \text{and} \; \O{E}_{\tilde{\O{S}}}\left[\mathbb{1}(\pi_{i,j} \geq \upsilon)\right] > 0\bigr\}.
\nonumber
\end{equation}
 In particular, if the base-level and hierarchical ideal designs are chosen to be uninformative, it follows that
 \begin{equation} \mathbb{\Pi}_{\O{E}}(\eta, \zeta, \upsilon) \xrightarrow{\eta \rightarrow \infty, \zeta \rightarrow \infty} \bigl\{\pi_{i,j}, \; (i,j) \in [\![m]\!] \times [\![n]\!]\; | \; \O{E}_{\mathcal U}\left[\mathbb{1}(\pi_{i,j} \geq \upsilon)\right] > 0 \bigr\}.
 \nonumber
\end{equation}
\item In the case of crisp marginals (i.e.\ no marginal uncertainties), $\mathbb{\Pi}_{\O{E}}$ contracts to a subset of  $\mathbb{\Pi}(\mu_{0}, \nu_{0})$ (\ref{eq:Kset}): 
\begin{eqnarray}
    \mathbb{\Pi}_{\O{E}}(\eta, \zeta, \upsilon) \xrightarrow{\eta \rightarrow 0, \zeta \rightarrow 0}&& \bigl\{\pi_{i,j}, \; (i,j) \in [\![m]\!] \times [\![n]\!] \; | \; \pi \in \mathbb{\Pi}(\mu_{0}, \nu_{0}) \; \text{and} \; \O{E}_{\O{S}^{o}}\left[\mathbb{1}(\pi_{i,j} \geq \upsilon)\right] > 0\bigr\} \nonumber\\
    &&\nonumber \\
    &&\subset \mathbb{\Pi}(\mu_{0}, \nu_{0}). \nonumber
\end{eqnarray}
\end{itemize}
We use $\mathbb{\Pi}_{\O{E}}(\eta, \zeta, \upsilon)$ to introduce our definition of fairness for contracts.  

\begin{definition}{\textbf{Fairness for contracts via diversified transport plans}}\label{def:fairness_def2} \\
    A random transport plan, $\pi\sim\O{S}^o (\pi | K )$, achieves fairness for contracts if it distributes the transport burden over all eligible contracts in $\mathbb{\Pi}_{\O{E}}(\eta, \zeta, \upsilon)$.
\end{definition}
For the purpose of illustration, we fix the optimal Kantorovitch potentials (\ref{eq:optKant}) to arbitrarily small values: $\lambda^o_{1}=\lambda^o_{2}=0.05$ (or, equivalently, large uncertainty radii $(\eta, \zeta)$), and the activation threshold to $\upsilon = 2\times 10^{-2}$. Both the base-level and hierarchical ideal designs (\ref{eq:joint_ideal_design}) are chosen to be uniform. We generate a sequence of relative frequency maps, each providing estimates of the probabilities that the respective contracts, $\pi_{i,j} \in \mathbb{\Pi}_{\O{E}}(\eta, \zeta, \upsilon)$, are active. We compare these to the base-level EOT matching policy (Figure \ref{fig:fig120}), which -- being oblivious to the uncertainty in the marginals -- yields a sparse transport policy and thus fails to achieve fairness for contracts (Definition \ref{def:fairness_def2}). In contrast, the random HFPD-OT matching policies enable a greater diversity by ensuring that more of the contracts are active, as shown in Figure \ref{fig:fig121}, \ref{fig:fig122} and \ref{fig:fig123}. These are the estimated activation probability maps, averaged over $N \in \{10, 50, 100\}$ randomly sampled transport plans,  $\pi^{(i)}\sim\O{S}^o (\pi | K )$, for $i \in [\![{N}]\!]$. As $N\rightarrow \infty$, these activation estimates converge to the ergodic limit, in which all eligible contracts have the same probability of being active.

\begin{figure}
    \centering
    \begin{subfigure}[b]{0.22\textwidth}
    \captionsetup{justification=centering}
        \includegraphics[width=\textwidth]{ 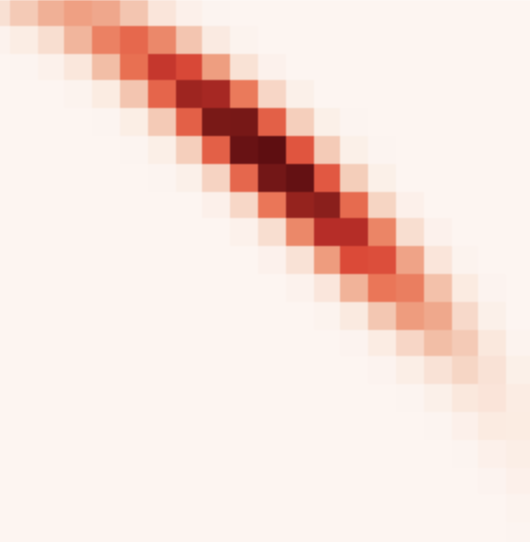}
        \caption{EOT plan computed between nominal marginals $\mu_{0}$ and $\nu_{0}$.}
        \label{fig:fig120}
    \end{subfigure} \hfill
    \begin{subfigure}[b]{0.22\textwidth}
    \captionsetup{justification=centering}
        \includegraphics[width=\textwidth]{ 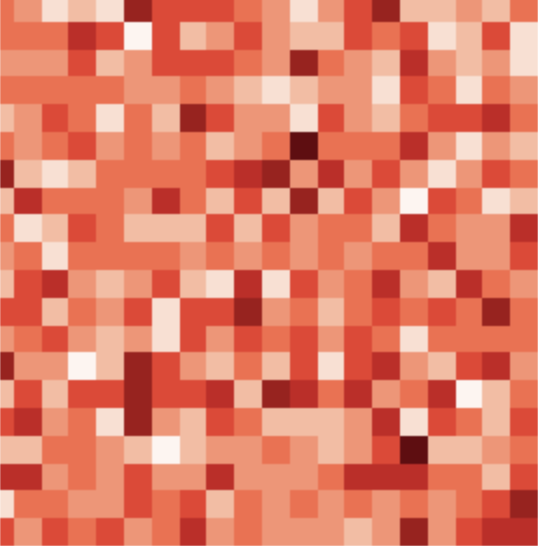}
        \caption{Estimated probabilities of active contracts,   $N=10$ realized plans.}
        \label{fig:fig121}
    \end{subfigure} \hfill
    \begin{subfigure}[b]{0.22\textwidth}
    \captionsetup{justification=centering}
        \includegraphics[width=\textwidth]{ 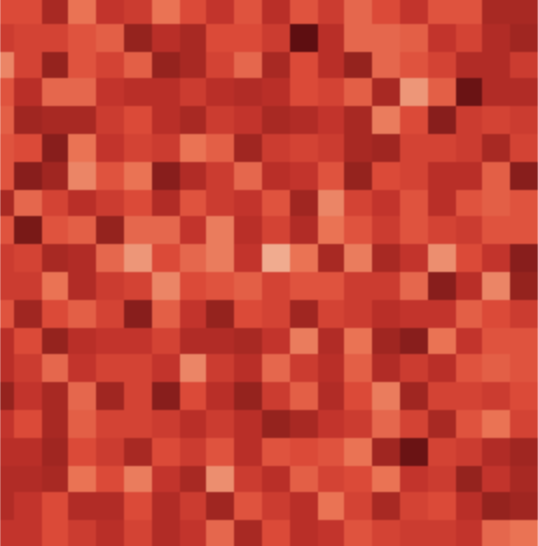}
        \caption{ Estimated probabilities of active contracts,   $N=50$ realized plans.}
        \label{fig:fig122}
    \end{subfigure} \hfill
    \begin{subfigure}[b]{0.285\textwidth}
    \captionsetup{justification=centering}
        \includegraphics[width=\textwidth]{ 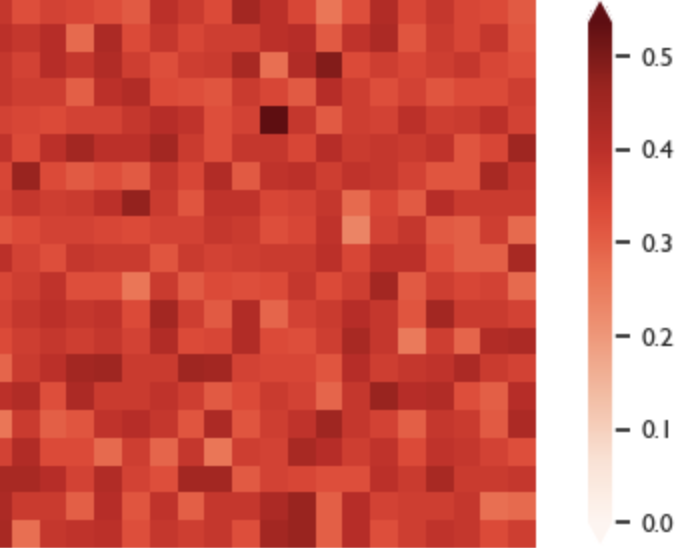}
        \caption{ Estimated probabilities of active contracts,   $N=100$ realized plans.}
        \label{fig:fig123}
    \end{subfigure}
    \caption{Comparison of the diversity of contracts induced by the conventional  base-level EOT solution vs HFPD-OT. \textbf{Figure (a)}: The base-level EOT plan, with the smoothness parameter fixed to $\epsilon = 10^{-3}$, induces a sparse policy, and therefore does not fairly distribute the  burden of transport across all eligible contracts, $\pi_{i,j}$. \textbf{Figures (b), (c), (d)}: random HFPD-OT policies induce  a long-term (i.e.\ ergodically) fair regime, where the  burden of transport is distributed across a larger set of contracts. Each entry in the relative frequency maps shows the estimated probability of activation of the corresponding contract, $\pi_{i,j}$. In the limit of $N\rightarrow \infty$  randomly realized matching policies, $\pi^{(i)}\sim\O{S}^o (\pi | K ) $, the map of estimated probabilities of active policies converges to a fair regime, where all eligible contracts  equally support the transport burden.}\label{fig:fig12}
\end{figure}

\begin{remark}\label{remark:randomness}
  Another way to appreciate fairness for contracts induced by randomized HFPD-OT plans is to study the {\em random\/} marginal cost,     $c_{i,j}$, associated with the contract $\pi_{i,j}$ (Figure \ref{fig:fig22}):
\begin{equation}
\label{eq:condcost}
    c_{i,j} \equiv \pi_{i, j} \O{C}_{i,j} \;\;, \;\; \pi \sim \mathsf{S}^{o}
\end{equation}
Recall that the  squared  2-Kantorovitch distance, 
\begin{equation}
\label{eq:2Wass}
\mathsf{KD}_{2}^{2}(\mu_{0}, \nu_{0}) \equiv\min_{\pi \in \mathbb{\Pi}(\mu_{0}, \nu_{0})}\sum_{i,j}\mathsf{C}_{i,j}\pi_{i,j},
\end{equation}
is the minimum expected transport cost between $\mu_{0}$ and $\nu_{0}$, for the Euclidean cost function, $\mathsf{C}$ (Section~\ref{sec:sim}) \citep{Villani2008OptimalTO}.
The base-level OT objective in \eqref{eq:2Wass} yields a fixed optimal solution, where the cost $c_{i,j}$ is immutable. Consequently,  the transport burden is supported by the same set of contracts. Let $\pi_{i_{0}, j_{0}}$ be one such contract where: 
\begin{equation}\label{eq:assymetric_burden_ot}
c_{i_{0}, j_{0}} >  \mathsf{KD}_{2}^{2}(\mu_{0}, \nu_{0}), \;\; (i_{0}, j_{0}) \in [\![m]\!] \times [\![n]\!]
\end{equation}
On the other hand, in HFPD-OT, and by virtue of the random nature of $c_{i_{0}, j_{0}}$, we can write the following Markov inequality: 
\begin{equation}\label{eq:markov}
\Pr \left[ c_{i_{0}, j_{0}} \geq  \mathsf{KD}_{2}^{2}(\mu_{0}, \nu_{0}) \right] \leq  \frac{\O{E}_{\pi \sim \mathsf{S}^{o}}\left[c_{i_{0}, j_{0}}\right]}{\mathsf{KD}_{2}^{2}(\mu_{0}, \nu_{0})}
\end{equation} 
Hence, this probability upper bound  depends on the ratio of the 
expected marginal transport cost associated with the contract, $\pi_{i_{0}, j_{0}}$ (\ref{eq:condcost}), to the squared 2-Kantorovitch distance between the nominal marginals \eqref{eq:2Wass}. Essentially, it provides an upper bound on the probability of a fairness-related proposition (Definition~\ref{def:fairness_def2}). Insights such as these may be used to establish operating conditions that are conducive to fairness. Such statistical handles on transport fairness are, of course, unavailable in conventional base-level OT.


\end{remark}

\section{Conclusions and next steps}\label{conclusions}
This paper recasts the optimal transport problem into a broader class of fully probabilistic design and generalized Bayesian inference techniques. In this new formalism, the transport plan is no longer regarded as a crisp, deterministic object, but is modeled as a random (i.e.\ uncertain) distribution in a hierarchical Bayesian setting. This is in clear contrast with the existing, certainty-equivalence-based OT paradigm. In this way, we augment the conventional base-level (i.e.\ deterministic) OT framework with the necessary tools to reason about uncertainty and design robust transport algorithms. In this new hierarchical setting, the object of interest is no longer the optimal transport plan, which may not even exist---since the marginals are themselves noisy, uncertain realizations of some underlying stochastic process---but is rather the optimal hyperprior, which is effectively a generative model over the set of transport plans. \\\\ 
We now recall  some key results on HFPD-OT, obtained in this paper:
\begin{itemize}
\item The functional form of the optimal hyperprior has been characterized in both the non-parametric and parametric settings. Importantly, we proved that the HFPD-OT setting is a generalization of the classical EOT in that the optimal transport plan can be recovered asymptotically when uncertainty in the marginals decreases.
\item Considering the parametric setting, we proposed an algorithm to approximate the Kantorovitch potentials and described some of the inferential properties of the hyperprior, highlighting its shape and location parameters. 
\item To illustrate the importance of HFPD-OT, we studied the problem of algorithmic fairness as it arises in fair market matching problems. First, we explored the role of randomization and diversification in eliciting fairer transport policies for agents, that is, for specific categories of workers and the companies which need their skills. Second, we investigated the role of randomization in eliciting fair matching policies for individual contracts between agents, by allowing the distribution of the transport burden across a larger set of contracts.
\end{itemize}
There remain important open questions to be studied and improvements to be implemented in subsequent work. The stochastic algorithms leveraged here enable a first approximation of the optimal hyperprior, but better samplers can be derived. Interestingly, sampling from the hyperprior may require new MCMC  techniques that leverage the geometry of the support of $\O{S}^o(\pi | K)$.  Moreover, the HFPD-OT application covered in this paper is on algorithmic fairness, however, we contend that the set of possible applications is broader: randomized policies play indeed an important role in  a diversity of problems related to generalizability and robustness in machine learning. 
Finally, a notable contribution of this paper has been to expand duality results from the classical setting in OT to the hierarchical framework of HFPD-OT. However, key theoretical results in  base-level deterministic OT---mainly those related to its geometry (\citep{Gangbo_McCann}, \citep{Villani2008OptimalTO}, {\em etc}.)---need  careful consideration within the extended framework of HFPD-OT.

\section*{Acknowledgement}
This work has been supported by the European Union’s Horizon Europe research and innovation programme, under grant agreement no.\ 101070568. It has also been supported by the European Union's competitive HORIZON-MSCA-2021-DN-01 (Marie Sklodowska-Curie Doctoral Networks) programme, under grant agreement no.\ 101073508. The authors also acknowledge the support of Innovate UK in underwriting both of the above grants, under the Horizon Europe Guarantee.

\section{Appendix: proof of strong duality in 
Theorem~\ref{main_theorem} (step 3)} \label{appendix}
The following  additional mathematical definitions are required, supplementing the preliminaries in Section~\ref{conventions}. 
\begin{itemize}
\item Besides being compact, we assume henceforth that $\mathbb{\Omega}_{X}$ and $\mathbb{\Omega}_{Y}$ are Hausdorff sets. This separability property guarantees uniqueness of limits and sequences. 
\item From compactness of $\mathbb{\Omega}_{X}$ and $\mathbb{\Omega}_{Y}$, it follows, by the Riesz-Markov-Kakutani Theorem \citep{folland}, that the topological dual of $\mathbb{C}(\mathbb{\Omega}_{X}\times\mathbb{\Omega}_{Y})$---the set of continuous functions on $\mathbb{\Omega}_{X}\times\mathbb{\Omega}_{Y}$---is the set of Radon measures with support in $\mathbb{\Omega}_{X} \times \mathbb{\Omega}_{Y}$. This also implies that $\mathbb{C}(\mathbb{\Omega}_{X}\times\mathbb{\Omega}_{Y})$ is a Banach space. Thus, by the Banach-Alaoglu Theorem, $\mathbbm{P}(\mathbb{\Omega}_{X}\times\mathbb{\Omega}_{Y})$ is compact in the weak-* topology \citep{Bill86}.
\item The previous compactness result allows us to again invoke the Riesz-Markov-Kakutani representation Theorem, which states that the topological dual of $\mathbb{C}(\mathbbm{P}(\mathbb{\Omega}_{X} \times \mathbb{\Omega}_{Y}))$ is the hierarchical space of Radon measures with support in $\mathbbm{P}(\mathbb{\Omega}_{X} \times \mathbb{\Omega}_{Y})$. We denote this dual space by $\mathbb{S}$.
The canonical duality pairing reads as follows \citep{folland}: \begin{equation}
    <f, \mathsf{S}> \equiv\int_{\mathbbm{P}(\mathbb{\Omega}_{X} \times \mathbb{\Omega}_{Y})} \hspace*{-1.1cm} fd\mathsf{S}
\end{equation} 
with $f \in \mathbb{C}(\mathbbm{P}(\mathbb{\Omega}_{X} \times \mathbb{\Omega}_{Y}))$ and $\mathsf{S} \in \mathbb{S}$. Later in the proof, we will constrain $\mathbb{S}$ to the set of hierarchical (probability) distributions.

\item If $\mathsf{O}: \mathbb{S} \rightarrow \mathbbm{R}^{p}$ is a linear map, its adjoint is defined as:  $\mathsf{O}^{*} \colon \mathbbm{R}^{p} \rightarrow \mathbb{C}(\mathbb{P}(\mathbb{\Omega}_{X} \times \mathbb{\Omega}_{Y}))$ such that: \begin{equation}\label{eq:linear_map} <\mathsf{O}(\mathsf{S}), z> = <\mathsf{S}, \mathsf{O}^{*}(z)> \end{equation} for $\mathsf{S} \in \mathbb{S}$ and $z \in \mathbb{R}^{p}$.

\item $f^{*}$ denotes the Legendre-Fenchel transform of $f$ defined in $\mathbb{C}(\mathbbm{P}(\mathbb{\Omega}_{X} \times \mathbb{\Omega}_{Y}))$. It is given by:
\begin{equation}\label{eq:lf_transform}
    f^{*}(u) \equiv\sup_{v \in \mathbb{C}(\mathbbm{P}(\mathbb{\Omega}_{X} \times \mathbb{\Omega}_{Y}))}(<u, v> - f(v))
\end{equation}
\item $\O{dom}(h)$ denotes the effective domain of the function $h \in \mathbb{C}(\mathbb{P}(\mathbb{\Omega}_{X}\times \mathbb{\Omega}_{Y}))$, defined as: $\O{dom}(h) \equiv\{\pi \in \mathbb{P}(\mathbb{\Omega}_{X}\times \mathbb{\Omega}_{Y}) \; | \; h(\pi) < \infty\}$.
\item Our proof relies on the notion of \textit{decomposable} spaces, as originally defined in Theorem 1 of \citep{rockafellar_integrals}. A space is decomposable if it is stable under bounded alterations over sets of finite measure. 
\item Let $\mathbb{L}(\mathbb{P}(\mathbb{\Omega}_{X}\times \mathbb{\Omega}_{Y}))$ denote the set of integrable functions, defined in $\mathbb{P}(\mathbb{\Omega}_{X}\times \mathbb{\Omega}_{Y})$. $\mathbb{L}(\mathbb{P}(\mathbb{\Omega}_{X}\times \mathbb{\Omega}_{Y}))$ is decomposable, since it satisfies the following conditions \citep{rockafellar_integrals}:
\begin{itemize}
\item $\mathbb{L}(\mathbb{P}(\mathbb{\Omega}_{X}\times \mathbb{\Omega}_{Y}))$ contains every bounded and measurable functions defined on $\mathbb{P}(\mathbb{\Omega}_{X}\times \mathbb{\Omega}_{Y})$.
\item If $h \in \mathbb{L}(\mathbb{P}(\mathbb{\Omega}_{X}\times \mathbb{\Omega}_{Y}))$ and $\mathbb{l} \in \mathcal{F}_{\mathbb{\Omega}_{\O{H}}}$ is an arbitrary set of finite measure in $\mathbb{P}(\mathbb{\Omega}_{X}\times \mathbb{\Omega}_{Y})$ \eqref{HFPD-OT}, then $\mathbb{L}(\mathbb{P}(\mathbb{\Omega}_{X}\times \mathbb{\Omega}_{Y}))$ contains $\chi_{\mathbb{l}} \boldsymbol{\cdot} h$, where $\boldsymbol{\cdot}$ denotes the dot product between the indicator function $\chi_{\mathbb{l}}$ of $\mathbb{l}$ and the function $h$.
\end{itemize}
\item The characteristic function of a (convex) set $\mathbb{A}$ is the convex function: \[
\mathbb{1}_{\mathbb{A}}(x) \equiv
\left\{
\begin{aligned}
0  & \;\; \text{if} \;\; x \in \mathbb{A},  \\
+ \infty & \;\; \text{otherwise.} \\
\end{aligned} \right.
\]
\end{itemize}
For the sake of completeness, we recall the main duality Theorem \citep{TyrrellRockafellar1974ConjugateDA} in the general setting, before specializing it to our problem later in the proof: 
\begin{theorem}[Fenchel-Rockafellar]\label{theo:fenchel_rockaff}
Let $(E, E^{*})$ and $(F, F^{*})$ be two topologically paired spaces. Let $\O{O} \colon E \rightarrow F$ be a continuous linear operator and $\O{O}^{*} \colon F^{*} \rightarrow E^{*}$ its adjoint. Let f and g be two lower semi-continuous and proper convex functions defined on E and F, respectively. If the following qualification condition is satisfied:
$\exists \; y^{*} \in dom(g^{*})$ \; s.t. \; $f^{*}$ is continuous at $\O{O}^{*}(y^{*})$,
then:
\begin{equation}\label{eq:strong_duality_theorem}
\max_{x \in E} -f(-x) - g(\O{O}(x)) = \min_{y^{*} \in F^{*}} f^{*}(\O{O}^{*}(y^{*})) + g^{*}(y^{*})
\end{equation}
\end{theorem}

\begin{proof}
Let's consider the Primal problem in \eqref{eq:primal}. Using Fubini's Theorem and the Bayesian hierarchical modelling consistency condition stated in \eqref{eq:expected_plan}, it is easy to show that this original problem can be formulated equivalently, over the set of hyperpriors $\mathbb{S}$, as follows:
\[
    (P): \;\;\;\;\; \O{S}^{o} \in \argmin_{\O{S} \in \mathbb{S}} \left\{\mathsf{D}_{\mathsf{KL}}\bigl(\mathsf{S}(\pi|K)|| \tilde{\mathsf{S}}(\pi|K)\bigr) \right\}
\]
subject to:
\[\left\{
\begin{aligned}
    \O{E}_{\mathsf{S}}(\mathsf{D}_{\mathsf{KL}}(\mu||\mu_{0})) \leq \eta  \\
 \O{E}_{\mathsf{S}}(\mathsf{D}_{\mathsf{KL}}(\nu||\nu_{0})) \leq \zeta 
\end{aligned} \right.\]
where $\tilde{\O{{S}}}$ is defined in  \eqref{eq:new_ideal_design}.
The constraints involve the following linear map:
\[
\mathsf{I}(\mathsf{S}) \equiv\int_{\mathbbm{P}(\mathbb{\Omega}_{X} \times \mathbb{\Omega}_{Y})} \mathsf{S}(\pi|K)d\mathcal{L}(\pi)
\]
besides our usual moment constraints: 
\[
 \mathsf{O}_{1}(\mathsf{S}) \equiv
      \O{E}_{\mathsf{S}}(\mathsf{D}_{\mathsf{KL}}(\mu||\mu_{0})) \;\;\; ,  \;\; \mathsf{O}_{2}(\mathsf{S}) \equiv
      \O{E}_{\mathsf{S}}(\mathsf{D}_{\mathsf{KL}}(\nu||\nu_{0}))
\]

For convenience, we denote by $\bar{\mathsf{O}}$ the linear map given by: \begin{equation}\label{eq:linear_map}
    \bar{\mathsf{O}}(\mathsf{S}) \equiv(\mathsf{I}(\mathsf{S}), \mathsf{O}_{1}(\mathsf{S}), \mathsf{O}_{2}(\mathsf{S})) \in \mathbbm{R}^{3}
\end{equation}
As usual, we can use the characteristic function to encode the constraints directly in the objective $(P)$ , yielding the following equivalent unconstrained problem:
\begin{equation}\label{eq:unconstrained_primal}
(P^{\prime}): \;\;\;  \O{S}^{o} \in \argmin_{\O{S}} \left\{\mathsf{D}_{\mathsf{KL}}\bigl(\mathsf{S}(\pi|K)|| \tilde{\mathsf{S}}(\pi|K)\bigr) + g_{0}(\bar{\mathsf{O}}(\mathsf{S})) \right\}
\end{equation}
where we define $g_{0}$ as follows:
\[
g_{0}(z_{0}, z_{1}, z_{2}) \equiv\mathbb{1}_{[0, \eta]}(z_{0}) + \mathbb{1}_{[0, \zeta]}(z_{1}) + \mathbb{1}_{\{1\}}(z_{2}) \;\; , \;\; (z_{0}, z_{1}, z_{2}) \in \mathbbm{R}^{3}
\]
We begin by deriving the Legendre-Fenchel dual of $\bar{\O{O}}(\cdot)$, $g_{0}(\cdot)$ and $\O{D}_{\O{KL}}(\cdot||\cdot)$, respectively.
By the definition of the adjoint in \eqref{eq:linear_map}, it is straightforward to show that $\mathsf{\bar{O}}^{*}$ is given by:
\[
\mathsf{\bar{O}}^{*}(\lambda_{1}, \lambda_{2}, \lambda_{3}) = \lambda_{1} \mathsf{D}_{\mathsf{KL}}(\mu||\mu_{0}) + \lambda_{2} \mathsf{D}_{\mathsf{KL}}(\nu||\nu_{0}) + \lambda_{3} \;\; , \;\; (\lambda_{1}, \lambda_{2}, \lambda_{3}) \in \mathbbm{R}^{3}
\]
Moreover, applying the definition of Legendre-Fenchel transform \eqref{eq:lf_transform} yields the following conjugate of $g_{0}$:
\[
\begin{split}
g_{0}^{*}(\lambda_{1}, \lambda_{2}, \lambda_{3}) = \lambda_{1}\eta + \lambda_{2} \zeta + \lambda_{3}
\end{split}
\]
We now turn our attention to the conjugate of $\mathsf{D}_{\mathsf{KL}}(\cdot||\cdot)$. To this aim, we first consider the following integral functional \citep{rockafellar_integrals}: 
\begin{align}
    \mathcal{I}_{\O{f}}(u)\colon \mathbb{C}
    (\mathbbm{P}(\mathbb{\Omega}_{X} \times \mathbb{\Omega}_{Y})) & \longrightarrow \mathbbm{R} \\
    u & \longmapsto \int_{\mathbbm{P}(\mathbb{\Omega}_{X} \times \mathbb{\Omega}_{Y})} \O{f}(\pi, u(\pi))d\mathcal{L}(\pi)
\end{align}
where: 
\[
\O{f}(\pi, u(\pi)) \equiv\tilde{\O{S}}(\pi|K) \exp\bigl(u(\pi) -1\bigr)
\]
$\O{f}(\pi, \cdot)$ is clearly an integrable, proper and convex function. As we saw earlier, the space $\mathbb{L}
    (\mathbbm{P}(\mathbb{\Omega}_{X} \times \mathbb{\Omega}_{Y}))$ is decomposable. Therefore, by Theorem 2 in \citep{rockafellar_integrals}, we can perform the Legendre-Fenchel transform of $\mathcal{I}_{\O{f}}$ through the integral sign and write:
\begin{equation}  \label{LF_Integral}
      \mathcal{I}^{*}_{\O{f}}(u) 
      = \mathcal{I}_{\O{f}^{*}}(\mathsf{S}) \equiv \int_{\mathbbm{P}(\mathbb{\Omega}_{X} \times \mathbb{\Omega}_{Y})} \O{f}^{*}(\pi, u(\pi))d\mathcal{L}(\pi)
\end{equation}
$\O{f}^{*}$ is obtained using again the definition of the Fenchel-Rockafellar transform \eqref{eq:lf_transform}:
\[
\begin{split}
\O{f}^{*}(\pi, \mathsf{S}(\pi|K)) & \equiv\sup_{u \in \mathbb{C}(\mathbbm{P}(\mathbb{\Omega}_{X} \times \mathbb{\Omega}_{Y}))}\bigl\{u(\pi)\mathsf{S}(\pi|K) - \O{f}(\pi, u(\pi))\bigr\} \\ &= \mathsf{S}(\pi|K)\log\Bigl(\frac{\mathsf{S}(\pi|K)}{\tilde{\O{S}}(\pi|K)}\Bigr)
\end{split}
\]
It follows that:
\[
\mathcal{I}^{*}_{\O{f}}(x) = \O{D}_{\O{KL}}(\mathsf{S}||\tilde{\mathsf{S}})
\]
There exists at least one hyperprior $\O{S} \in \mathbb{S}$ s.t. $\O{f}^{*}$ is an integrable function of $\pi$ (consider for instance $\O{S} = \O{\tilde{S}}$). It follows, by Theorem 1 in \citep{rockafellar_integrals}, that $\mathcal{I}_{\O{f}}$ is a well-defined convex functional. Thus, the conjugacy operator acts as an involution on $\mathcal{I}_{\O{f}}$, yielding:
\[
\O{D}_{\O{KL}}^{*} = \mathcal{I}_{\O{f}}^{**} =  \mathcal{I}_{\O{f}}
\]
Going back to our main Theorem in \eqref{theo:fenchel_rockaff}, it is obvious that $\O{D}_{\mathsf{KL}}(\cdot||\cdot)$ and $g_{0}(\cdot)$ are lower semicontinuous, proper and convex. Furthermore, $\O{D}_{\mathsf{KL}}^{*}(\cdot||\cdot)$ is continuous everywhere \textit{w.r.t} the uniform norm (Theorem 4 in \citep{TyrrellRockafellar1974ConjugateDA}). It follows that strong duality holds and that the primal and dual problems are equal, the dual reading as follows:
\begin{equation}\label{dual_final}
\begin{split}
    (D): \;\;\;\; \sup_{(\lambda_{1}, \lambda_{2}, \lambda_{3}) \in \mathbbm{R}^{3}} \left\{-\int_{\mathbbm{P}(\mathbb{\Omega}_{X} \times \mathbb{\Omega}_{Y})} \tilde{\O{S}}(\pi|K) \exp\bigl(\bar{\O{O}}^{*}(-\lambda_{1}, -\lambda_{2}, -\lambda_{3})-1\bigr)d\mathcal{L}(\pi) -\lambda_{1}\eta -\lambda_{2}\zeta - \lambda_{3} \right\}
\end{split}
\end{equation}

One can simplify further the previous result by maximizing \eqref{dual_final} \textit{w.r.t} $\lambda_{3}$ for fixed $(\lambda_{1}, \lambda_{2})$, yielding the following value for $\lambda^{o}_{3}$:
\begin{equation}
\lambda^{o}_{3} = \log\left(\int_{\mathbbm{P}(\mathbb{\Omega}_{X} \times \mathbb{\Omega}_{Y})} \tilde{\O{S}}(\pi|K) \exp\left(-\lambda_{1} \mathsf{D}_{\mathsf{KL}}(\mu||\mu_{0}) -\lambda_{2} \mathsf{D}_{\mathsf{KL}}(\nu||\nu_{0})-1\right)d\mathcal{L}(\pi)\right)
\end{equation}
By substituting $\lambda^{o}_{3}$ back into \eqref{dual_final}, we obtain \eqref{eq:dual}. \\\\
The optimality condition: 
\[
0 \in \partial \O{D}_{\O{KL}}(\mathsf{S}(\pi|K)) + \partial{g_{0}}\bigl(\bar{O}(\mathsf{S}\bigr))
\]
implies that the primal and dual optimal solutions should satisfy the following extremality conditions  \citep{pjm_1102992608}: 
\[\left\{
\begin{aligned}
        \mathsf{S}^{o} \in \partial{\mathcal{I}_{\O{f}}(-\O{O}^{*}(\lambda_{1}, \lambda_{2}))} \;\;\; \mathcal{L}-a.e. \\
         (-\lambda_{1}^{o},  -\lambda_{2}^{o}) \in \partial{g_{0}\bigl(\O{O}_{1}(\O{S}), \O{O}_{2}(\O{S}) \bigr)}
\end{aligned} \right.\]
$\mathcal{I}_{\O{f}}$ being differentiable everywhere, its sub-differential reduces to the usual gradient, leading to the same optimal hyperprior derived earlier using information processing arguments \eqref{eq:hyperprior}:
\[
\mathsf{S}^{o} \propto  \exp\left(-\lambda_{1}^{o} \mathsf{D}_{\mathsf{KL}}(\mu||\mu_{0})\right) \tilde{\O{S}}(\pi|K) \exp\left(-\lambda_{2}^{o} \mathsf{D}_{\mathsf{KL}}(\nu||\nu_{0})\right) \;\;\; \mathcal{L}-a.e.
\]
On the other hand, noting that the sub-differential of the indicator function $g_{0}$ is the normal cone $\bar{\mathsf{N}}_{\mathbb{Q}}\bigl(\O{O}_{1}(\O{S}), \O{O}_{2}(\O{S})\bigr)$, defined as follows:
\begin{equation}\label{diff_two}
    \bar{\mathsf{N}}_{\mathbb{Q}}\bigl(\O{O}_{1}(\O{S}), \O{O}_{2}(\O{S})\bigr) \equiv\left\{v \in \mathbbm{R}^{2} \;\;|\;\; v^{\intercal}\left[\boldsymbol{x} - \begin{bmatrix} \O{O}_{1}(\O{S}) \\  \O{O}_{2}(\O{S}) \end{bmatrix}\right] \preceq 0, \forall \boldsymbol{x} \in [0, \eta] \times [0, \zeta]\right\}
\end{equation}
the following optimality conditions are obtained, for the special choice of $\boldsymbol{x} = (\eta, \zeta)$ plugged in \eqref{diff_two}:
\[\left\{
\begin{aligned}
        \lambda_{1}^{o}\Bigl(\eta - \O{E}_{\mathsf{S}}\bigl(\mathsf{D}_{\mathsf{KL}}(\mu||\mu_{0}\bigr)\Bigr) \geq 0 \\
        \lambda_{2}^{o}\Bigl(\zeta - \O{E}_{\mathsf{S}}\bigl(\mathsf{D}_{\mathsf{KL}}(\nu||\nu_{0}\bigr)\Bigr) \geq 0
\end{aligned} \right.\]
Thus:
$\boldsymbol{\lambda} \equiv(\lambda_{1}^{o}, \lambda_{2}^{o}) \succeq 0$. 
\end{proof}

\bibliography{references.bib}

\begin{thebibliography}{49}
\providecommand{\natexlab}[1]{#1}
\providecommand{\url}[1]{\texttt{#1}}
\expandafter\ifx\csname urlstyle\endcsname\relax
  \providecommand{\doi}[1]{doi: #1}\else
  \providecommand{\doi}{doi: \begingroup \urlstyle{rm}\Url}\fi

\bibitem[Arjovsky et~al.(2017)Arjovsky, Chintala, and Bottou]{arjovsky2017wasserstein}
Martin Arjovsky, Soumith Chintala, and L{\'e}on Bottou.
\newblock {W}asserstein generative adversarial networks.
\newblock In Doina Precup and Yee~Whye Teh, editors, \emph{Proceedings of the 34th International Conference on Machine Learning}, volume~70 of \emph{Proceedings of Machine Learning Research}, pages 214--223. PMLR, 06--11 Aug 2017.

\bibitem[Courty et~al.(2017)Courty, Flamary, Tuia, and Rakotomamonjy]{OT_Domain_Adaptation}
Nicolas Courty, Rémi Flamary, Devis Tuia, and Alain Rakotomamonjy.
\newblock Optimal transport for domain adaptation.
\newblock \emph{IEEE Transactions on Pattern Analysis and Machine Intelligence}, 39\penalty0 (9):\penalty0 1853--1865, 2017.
\newblock \doi{10.1109/TPAMI.2016.2615921}.

\bibitem[Mathon et~al.(2014)Mathon, Cayre, Bas, and Macq]{OT_Watermarking}
Benjamin Mathon, François Cayre, Patrick Bas, and Benoit Macq.
\newblock Optimal transport for secure spread-spectrum watermarking of still images.
\newblock \emph{Image Processing, IEEE Transactions on}, 23:\penalty0 1694--1705, 04 2014.
\newblock \doi{10.1109/TIP.2014.2305873}.

\bibitem[Guerreiro et~al.(2023)Guerreiro, Colombo, Piantanida, and Martins]{guerreiro2023optimal}
Nuno~M. Guerreiro, Pierre Colombo, Pablo Piantanida, and André F.~T. Martins.
\newblock Optimal transport for unsupervised hallucination detection in neural machine translation, 2023.

\bibitem[Galichon(2016)]{galichon}
Alfred Galichon.
\newblock \emph{Optimal Transport Methods in Economics}.
\newblock Princeton University Press, 2016.

\bibitem[Saumier et~al.(2015)Saumier, Khouider, and Agueh]{OT_fluid_dynamics}
Louis-Philippe Saumier, Boualem Khouider, and Martial Agueh.
\newblock Optimal transport for particle image velocimetry: Real data and postprocessing algorithms.
\newblock \emph{SIAM Journal on Applied Mathematics}, 75\penalty0 (6):\penalty0 2495--2514, 2015.
\newblock ISSN 00361399.

\bibitem[{El Moselhy} and Marzouk(2012)]{ELMOSELHY20127815}
Tarek~A. {El Moselhy} and Youssef~M. Marzouk.
\newblock Bayesian inference with optimal maps.
\newblock \emph{Journal of Computational Physics}, 231\penalty0 (23):\penalty0 7815--7850, 2012.
\newblock ISSN 0021-9991.
\newblock \doi{https://doi.org/10.1016/j.jcp.2012.07.022}.

\bibitem[Villani(2008)]{Villani2008OptimalTO}
C{\'e}dric Villani.
\newblock \emph{Optimal Transport: Old and New}.
\newblock Springer, 2008.

\bibitem[Peyré and Cuturi(2019)]{MAL-073}
Gabriel Peyré and Marco Cuturi.
\newblock Computational optimal transport: With applications to data science.
\newblock \emph{Foundations and Trends® in Machine Learning}, 11\penalty0 (5-6):\penalty0 355--607, 2019.

\bibitem[Ben-Tal et~al.(2009)Ben-Tal, Ghaoui, and Nemirovski]{robust_optimization}
Aharon Ben-Tal, Laurent Ghaoui, and Arkadi Nemirovski.
\newblock \emph{Robust Optimization}.
\newblock 08 2009.
\newblock ISBN 9781400831050.
\newblock \doi{10.1515/9781400831050}.

\bibitem[Sklar(1959)]{sklar}
Abe Sklar.
\newblock Fonctions de r\'epartition \`{a} n dimensions et leurs marges.
\newblock pages 229--231. Publications de l’Institut de Statistique de l’Universit\'e de Paris, 8, 1959.

\bibitem[Goodman(1953)]{Goodman1953EcologicalRA}
Leo Goodman.
\newblock Ecological regressions and behavior of individuals.
\newblock \emph{American Sociological Review}, 18:\penalty0 663, 1953.

\bibitem[Wakefield(2004)]{wakefield_ecological_inference}
Jon Wakefield.
\newblock Ecological inference for 2x2 tables (with discussion).
\newblock \emph{Journal of the Royal Statistical Society Series A}, 167:\penalty0 385--445, 08 2004.
\newblock \doi{10.1111/j.1467-985x.2004.02046.x}.

\bibitem[Frogner and Poggio(2019)]{pmlr-v97-frogner19a}
Charlie Frogner and Tomaso Poggio.
\newblock Fast and flexible inference of joint distributions from their marginals.
\newblock In Kamalika Chaudhuri and Ruslan Salakhutdinov, editors, \emph{Proceedings of the 36th International Conference on Machine Learning}, volume~97 of \emph{Proceedings of Machine Learning Research}, pages 2002--2011. PMLR, 09--15 Jun 2019.

\bibitem[Mallasto et~al.(2021)Mallasto, Heinonen, and Kaski]{pmlr-v157-mallasto21a}
Anton Mallasto, Markus Heinonen, and Samuel Kaski.
\newblock Bayesian inference for optimal transport with stochastic cost.
\newblock In \emph{Proceedings of The 13th Asian Conference on Machine Learning}, volume 157 of \emph{Proceedings of Machine Learning Research}, pages 1601--1616. PMLR, 2021.

\bibitem[Kárný and Kroupa(2012)]{KARNY2012105}
Miroslav Kárný and Tomáš Kroupa.
\newblock Axiomatisation of fully probabilistic design.
\newblock \emph{Information Sciences}, 186\penalty0 (1):\penalty0 105--113, 2012.
\newblock ISSN 0020-0255.
\newblock \doi{https://doi.org/10.1016/j.ins.2011.09.018}.

\bibitem[Séjourné et~al.(2023)Séjourné, Peyré, and Vialard]{SEJOURNE2023407}
Thibault Séjourné, Gabriel Peyré, and François-Xavier Vialard.
\newblock Chapter 12 - unbalanced optimal transport, from theory to numerics.
\newblock In Emmanuel Trélat and Enrique Zuazua, editors, \emph{Numerical Control: Part B}, volume~24 of \emph{Handbook of Numerical Analysis}, pages 407--471. Elsevier, 2023.
\newblock \doi{https://doi.org/10.1016/bs.hna.2022.11.003}.

\bibitem[Cuturi(2013)]{cuturi2013sinkhorn}
Marco Cuturi.
\newblock Sinkhorn distances: Lightspeed computation of optimal transport.
\newblock In C.J. Burges, L.~Bottou, M.~Welling, Z.~Ghahramani, and K.Q. Weinberger, editors, \emph{Advances in Neural Information Processing Systems}, volume~26. Curran Associates, Inc., 2013.

\bibitem[Quinn et~al.(2025)Quinn, Boufelja~Yacobi, Corless, and Shorten]{sby2022fully}
Anthony Quinn, Sarah Boufelja~Yacobi, Martin Corless, and Robert Shorten.
\newblock Fully probabilistic design for optimal transport.
\newblock \emph{Communications in Optimization Theory}, 2025.
\newblock To appear.

\bibitem[Jeffreys(1939)]{Jeffreys1939-JEFTOP-5}
Harold Jeffreys.
\newblock \emph{Theory of Probability}.
\newblock Clarendon Press, Oxford, England, 1939.

\bibitem[Carlier et~al.(2017)Carlier, Duval, Peyr\'{e}, and Schmitzer]{carlier2017convergence}
Guillaume Carlier, Vincent Duval, Gabriel Peyr\'{e}, and Bernhard Schmitzer.
\newblock Convergence of entropic schemes for optimal transport and gradient flows.
\newblock \emph{SIAM Journal on Mathematical Analysis}, 49\penalty0 (2):\penalty0 1385--1418, 2017.
\newblock \doi{10.1137/15M1050264}.

\bibitem[Quinn et~al.(2016)Quinn, Kárný, and Guy]{QUINN2016532}
Anthony Quinn, Miroslav Kárný, and Tatiana~V. Guy.
\newblock Fully probabilistic design of hierarchical {B}ayesian models.
\newblock \emph{Information Sciences}, 369:\penalty0 532--547, 2016.
\newblock ISSN 0020-0255.
\newblock \doi{https://doi.org/10.1016/j.ins.2016.07.035}.

\bibitem[Bissiri et~al.(2016)Bissiri, Holmes, and Walker]{Bissiri_2016}
Pier~Giovanni Bissiri, Chris Holmes, and Stephen Walker.
\newblock A general framework for updating belief distributions.
\newblock \emph{Journal of the Royal Statistical Society Series B: Statistical Methodology}, 78\penalty0 (5):\penalty0 1103--1130, feb 2016.
\newblock \doi{10.1111/rssb.12158}.

\bibitem[Savage(1971)]{LeonardSaravePersoProbabilities}
Leonard~J. Savage.
\newblock Elicitation of personal probabilities and expectations.
\newblock \emph{Journal of the American Statistical Association}, 66\penalty0 (336):\penalty0 783--801, 1971.
\newblock \doi{10.1080/01621459.1971.10482346}.

\bibitem[Rockafellar(1967)]{pjm_1102992608}
Ralph~Tyrrell Rockafellar.
\newblock {Duality and stability in extremum problems involving convex functions.}
\newblock \emph{Pacific Journal of Mathematics}, 21\penalty0 (1):\penalty0 167 -- 187, 1967.

\bibitem[Krac{\'\i}k and K{\'a}rn{\'y}(2005)]{kracik2005merging}
Jan Krac{\'\i}k and Miroslav K{\'a}rn{\'y}.
\newblock Merging of data knowledge in {B}ayesian estimation.
\newblock In \emph{International Conference on Informatics in Control, Automation and Robotics}, volume~2, pages 229--232, 2005.

\bibitem[Kolmogorov and Sarmanov(1960)]{bernstein-von-mises}
Andrey~Nikolaevich Kolmogorov and Oleg~Vasil\'evich Sarmanov.
\newblock The work of {S}. {N}. {B}ernshtein on the theory of probability.
\newblock \emph{Theory of Probability \& Its Applications}, 5\penalty0 (2):\penalty0 197--203, 1960.
\newblock \doi{10.1137/1105017}.

\bibitem[Delahaye et~al.(2019)Delahaye, Chaimatanan, and Mongeau]{Delahaye2019}
Daniel Delahaye, Supatcha Chaimatanan, and Marcel Mongeau.
\newblock \emph{Simulated Annealing: From Basics to Applications}, pages 1--35.
\newblock Springer International Publishing, Cham, 2019.
\newblock ISBN 978-3-319-91086-4.
\newblock \doi{10.1007/978-3-319-91086-4_1}.

\bibitem[Quinn(2012)]{10.4108/icst.valuetools.2011.246122}
Anthony Quinn.
\newblock Recursive inference for inverse problems using variational {B}ayes methodology.
\newblock In \emph{1st International ICST Workshop on New Computational Methods for Inverse Problems}. ACM, 2012.

\bibitem[Nocedal and Wright(2006)]{NoceWrig06}
Jorge Nocedal and Stephen~J. Wright.
\newblock \emph{Numerical Optimization}.
\newblock Springer, New York, NY, USA, 2e edition, 2006.

\bibitem[Nesterov(2018)]{nesterov}
Yurii Nesterov.
\newblock \emph{Lectures on Convex Optimization}.
\newblock Springer Publishing Company, Incorporated, 2nd edition, 2018.
\newblock ISBN 3319915770.

\bibitem[Betancourt(2017)]{betancourt2018conceptual}
Michael Betancourt.
\newblock A conceptual introduction to {H}amiltonian {M}onte {C}arlo.
\newblock \emph{arXiv: Methodology}, 2017.

\bibitem[Mangoubi and Smith(2019)]{pmlr-v89-mangoubi19a}
Oren Mangoubi and Aaron Smith.
\newblock Mixing of {H}amiltonian {M}onte {C}arlo on strongly log-concave distributions 2: Numerical integrators.
\newblock In Kamalika Chaudhuri and Masashi Sugiyama, editors, \emph{Proceedings of the Twenty-Second International Conference on Artificial Intelligence and Statistics}, volume~89 of \emph{Proceedings of Machine Learning Research}, pages 586--595. PMLR, 16--18 Apr 2019.

\bibitem[Snyman(2005)]{SnymanGradientLineSearch}
Jan Snyman.
\newblock A gradient‐only line search method for the conjugate gradient method applied to constrained optimization problems with severe noise in the objective function.
\newblock \emph{International Journal for Numerical Methods in Engineering}, 62:\penalty0 72 -- 82, 01 2005.
\newblock \doi{10.1002/nme.1189}.

\bibitem[Chen and Vempala(2022)]{chen2019optimalconvergenceratehamiltonian}
Zongchen Chen and Santosh~S. Vempala.
\newblock Optimal convergence rate of {H}amiltonian {M}onte {C}arlo for strongly logconcave distributions.
\newblock \emph{Theory of Computing}, 18\penalty0 (9):\penalty0 1--18, 2022.
\newblock \doi{10.4086/toc.2022.v018a009}.

\bibitem[Barocas et~al.(2023)Barocas, Hardt, and Narayanan]{Barocas2018FairnessAM}
Solon Barocas, Moritz Hardt, and Arvind Narayanan.
\newblock \emph{Fairness and machine learning: Limitations and opportunities}.
\newblock MIT press, 2023.

\bibitem[Gordaliza et~al.(2019)Gordaliza, Barrio, Fabrice, and Loubes]{delbarrio2018obtaining}
Paula Gordaliza, Eustasio~Del Barrio, Gamboa Fabrice, and Jean-Michel Loubes.
\newblock Obtaining fairness using optimal transport theory.
\newblock In Kamalika Chaudhuri and Ruslan Salakhutdinov, editors, \emph{Proceedings of the 36th International Conference on Machine Learning}, volume~97 of \emph{Proceedings of Machine Learning Research}, pages 2357--2365. PMLR, 09--15 Jun 2019.

\bibitem[Hughes and Chen(2021)]{huguesjasonchen2021}
Jason Hughes and Juntao Chen.
\newblock Fair and distributed dynamic optimal transport for resource allocation over networks.
\newblock In \emph{2021 55th Annual Conference on Information Sciences and Systems (CISS)}, pages 1--6, 2021.
\newblock \doi{10.1109/CISS50987.2021.9400236}.

\bibitem[Galichon(2021)]{galichon2021unreasonableeffectivenessoptimaltransport}
Alfred Galichon.
\newblock The unreasonable effectiveness of optimal transport in economics.
\newblock \emph{arXiv preprint arXiv:2107.04700}, 2021.

\bibitem[Echenique et~al.(2024)Echenique, Root, and Sandomirskiy]{echenique2024stablematchingtransportation}
Federico Echenique, Joseph Root, and Fedor Sandomirskiy.
\newblock Stable matching as transportation.
\newblock In \emph{Proceedings of the 25th ACM Conference on Economics and Computation}, EC '24, page 418, New York, NY, USA, 2024. Association for Computing Machinery.
\newblock ISBN 9798400707049.
\newblock \doi{10.1145/3670865.3673585}.

\bibitem[Hoffman and Gelman(2011)]{hoffman2011nouturnsampleradaptivelysetting}
Matthew~D. Hoffman and Andrew Gelman.
\newblock The no-u-turn sampler: adaptively setting path lengths in hamiltonian monte carlo.
\newblock \emph{Journal of Machine Learning Research}, 15:\penalty0 1593--1623, 2011.

\bibitem[Flamary et~al.(2021)Flamary, Courty, Gramfort, Alaya, Boisbunon, Chambon, Chapel, Corenflos, Fatras, Fournier, Gautheron, Gayraud, Janati, Rakotomamonjy, Redko, Rolet, Schutz, Seguy, Sutherland, Tavenard, Tong, and Vayer]{10.5555/3546258.3546336}
R\'{e}mi Flamary, Nicolas Courty, Alexandre Gramfort, Mokhtar~Z. Alaya, Aur\'{e}lie Boisbunon, Stanislas Chambon, Laetitia Chapel, Adrien Corenflos, Kilian Fatras, Nemo Fournier, L\'{e}o Gautheron, Nathalie~T.H. Gayraud, Hicham Janati, Alain Rakotomamonjy, Ievgen Redko, Antoine Rolet, Antony Schutz, Vivien Seguy, Danica~J. Sutherland, Romain Tavenard, Alexander Tong, and Titouan Vayer.
\newblock {POT}: Python optimal transport.
\newblock \emph{J. Mach. Learn. Res.}, 22\penalty0 (1), January 2021.
\newblock ISSN 1532-4435.

\bibitem[Jelinek et~al.(1977)Jelinek, Mercer, Bahl, and Baker]{Jelinek1977PerplexityaMO}
Frederick Jelinek, Robert~L. Mercer, Lalit~R. Bahl, and Janet~M. Baker.
\newblock Perplexity—a measure of the difficulty of speech recognition tasks.
\newblock \emph{Journal of the Acoustical Society of America}, 62, 1977.

\bibitem[Martí et~al.(2022)Martí, Martínez-Gavara, Pérez-Peló, and Sánchez-Oro]{MARTI2022795}
Rafael Martí, Anna Martínez-Gavara, Sergio Pérez-Peló, and Jesús Sánchez-Oro.
\newblock A review on discrete diversity and dispersion maximization from an or perspective.
\newblock \emph{European Journal of Operational Research}, 299\penalty0 (3):\penalty0 795--813, 2022.
\newblock ISSN 0377-2217.
\newblock \doi{https://doi.org/10.1016/j.ejor.2021.07.044}.

\bibitem[Gangbo and McCann(1996)]{Gangbo_McCann}
Wilfrid Gangbo and Robert~J. McCann.
\newblock {The geometry of optimal transportation}.
\newblock \emph{Acta Mathematica}, 177\penalty0 (2):\penalty0 113 -- 161, 1996.
\newblock \doi{10.1007/BF02392620}.

\bibitem[Folland(1999)]{folland}
Gerald~Budge Folland.
\newblock \emph{Real Analysis : Modern Techniques and Their Applications}.
\newblock Wiley, New York, 1999.

\bibitem[Billingsley(1999)]{Bill86}
Patrick Billingsley.
\newblock \emph{Convergence of probability measures}.
\newblock Wiley Series in Probability and Statistics: Probability and Statistics. John Wiley \& Sons Inc., New York, second edition, 1999.
\newblock ISBN 0-471-19745-9.
\newblock A Wiley-Interscience Publication.

\bibitem[Rockafellar(1971)]{rockafellar_integrals}
Ralph~Tyrrell Rockafellar.
\newblock {Integrals which are convex functionals. II.}
\newblock \emph{Pacific Journal of Mathematics}, 39\penalty0 (2):\penalty0 439 -- 469, 1971.

\bibitem[Rockafellar(1974)]{TyrrellRockafellar1974ConjugateDA}
Ralph~Tyrrell Rockafellar.
\newblock Conjugate duality and optimization.
\newblock Society for Industrial and Applied Mathematics, 1974.

\end{thebibliography}
\end{document}